\documentclass[lettersize,journal]{IEEEtran}
\usepackage{amsmath,amsthm,amscd,amssymb,latexsym,upref,stmaryrd,cite}
\usepackage{setspace}
\usepackage{pgf}
\usepackage{tikz}
\usepackage{caption}
\usetikzlibrary{arrows, decorations.pathmorphing, backgrounds, positioning, fit, petri, automata}
\usepackage{multirow}  
\definecolor{white1}{rgb}{1,1,1}
\newcommand*{\mailto}[1]{\href{mailto:#1}{\nolinkurl{#1}}}

\usepackage{threeparttable}

\newcommand{\tr}{\text{\rm{Tr}}}    
\newcommand{\diag}{\text{\rm{diag}}}

\newcommand{\beq}{\begin{equation}}
\newcommand{\eeq}{\end{equation}}
\newcommand{\ba}{\begin{align}}
\newcommand{\ea}{\end{align}}

\renewcommand{\Re}{\text{\rm Re}}
\renewcommand{\Im}{\text{\rm Im}}

 \DeclareMathOperator{\Tr}{Tr}

 \allowdisplaybreaks

\newtheorem{theorem}{Theorem}    

\newtheorem{corollary}[theorem]{Corollary}

\theoremstyle{definition}

\def\IEEEkeywordsname{Index Terms}
\usepackage[numbers,sort&compress]{natbib}
\usepackage{mathtools}
\usepackage{booktabs}
\DeclareCaptionFont{tinyfont}{\footnotesize}
\usepackage[font={tinyfont}]{caption}
\begin{document}

\title{RIS-Aided Cooperative ISAC Networks for Structural Health Monitoring}

\author{Jie~Yang,~\IEEEmembership{Member,~IEEE,} \  Chao-Kai~Wen,~\IEEEmembership{Fellow,~IEEE,} \ Xiao~Li,~\IEEEmembership{Member,~IEEE,} and Shi~Jin,~\IEEEmembership{Fellow,~IEEE }
	\thanks{ 
		Jie~Yang and Shi~Jin are also with the Frontiers Science Center for Mobile Information Communication and Security, Southeast University, Nanjing 210096, China (e-mail: yangjie@seu.edu.cn).
		Jie~Yang is with the Key Laboratory of Measurement and Control of Complex Systems of Engineering, Ministry of Education, Southeast University, Nanjing 210096, China.
		Xiao~Li and Shi~Jin are with the National Mobile Communications Research Laboratory, Southeast University, Nanjing 210096, China (e-mail: li\_xiao@seu.edu.cn; jinshi@seu.edu.cn).
		Chao-Kai Wen is with the Institute of Communications Engineering, National Sun Yat-sen University, Kaohsiung 80424, Taiwan. (e-mail: chaokai.wen@mail.nsysu.edu.tw).
		Part of this study has been presented in WOCC 2024 \cite{wocc24}.}
}

\date{\today}
\maketitle

\begin{abstract}

Integrated sensing and communication (ISAC) is a key feature of future cellular systems, enabling applications such as intruder detection, monitoring, and tracking using the same infrastructure. However, its potential for structural health monitoring (SHM), which requires the detection of slow and subtle structural changes, remains largely unexplored due to challenges such as multipath interference and the need for ultra-high sensing precision. This study introduces a novel theoretical framework for SHM via ISAC by leveraging reconfigurable intelligent surfaces (RIS) as reference points in collaboration with base stations and users. By dynamically adjusting RIS phases to generate distinct radio signals that suppress background multipath interference, measurement accuracy at these reference points is enhanced. We theoretically analyze RIS-aided collaborative sensing in three-dimensional cellular networks using Fisher information theory, demonstrating how increasing observation time, incorporating additional receivers (even with self-positioning errors), optimizing RIS phases, and refining collaborative node selection can reduce the position error bound to meet SHM’s stringent accuracy requirements. Furthermore, we develop a Bayesian inference model to identify structural states and validate damage detection probabilities. Both theoretical and numerical analyses confirm ISAC's capability for millimeter-level deformation detection, highlighting its potential for high-precision SHM applications.

\begin{IEEEkeywords}
	 Integrated sensing and communication, structural health monitoring, reconfigurable intelligent surfaces, spatiotemporal cooperation, multi-user cooperation.
\end{IEEEkeywords}

\end{abstract}

\vspace{-0.5cm}
\section{Introduction}

As urbanization progresses, structural health monitoring (SHM) becomes increasingly critical. Changes in building geometry may compromise structural integrity, leading to cracks and potential collapses that endanger lives and property if left undetected. SHM employs sensors to continuously monitor building structures and track capacity changes due to aging or accidental damage \cite{SHM}. Traditional geodetic methods, including global navigation satellite systems (GNSS), laser scanning, depth cameras, and radar interferometry, play significant roles in SHM \cite{SAR}. However, these methods require frequent observations that are both time-consuming and labor-intensive. Moreover, GNSS is primarily limited to measuring top displacements. To overcome these limitations, sophisticated sensors such as accelerometers, hydrostatic levels, and inclinometers are deployed to detect displacements at multiple points. Nevertheless, challenges such as error accumulation, the absence of a standardized external coordinate system, and difficulties in efficient information fusion persist \cite{ZHU2025102614}.

Integrated sensing and communication (ISAC), which leverages the inherent functions of cellular systems for sensing, has recently emerged \cite{ISAC}, \cite{TenLu} and is recognized as one of the six typical application scenarios for the sixth-generation (6G) mobile communication \cite{6Ghenk}. 
The third generation partnership project (3GPP) discusses sensing services (e.g., autonomous/assisted driving, vehicle-to-everything, unmanned aerial vehicles, three-dimensional (3D) map reconstruction, smart city, smart home, factories, healthcare, and the maritime sector) and related key performance indicators \cite{3GPP}, \cite{3GPP22837_j40}.
Distributed intelligent ISAC has been proposed to expand collaborative capabilities \cite{DISAC}, and networked ISAC schemes involving multiple transceivers have been shown to enhance ISAC services \cite{CISAC}, \cite{COWei}, \cite{DCOWei}, \cite{COMoe}. In addition, cellular systems can promptly broadcast warning messages upon detecting hazards. These advances motivate us to exploit the dense, distributed, and collaborative characteristics of cellular systems for SHM, one of the key applications of interest for the 3GPP.

Unlike typical ISAC applications such as user localization and environmental reconstruction \cite{SLAM}, \cite{DISAChenk}, SHM imposes relatively low real-time requirements on sensing, enabling operation during idle periods of communication nodes. However, SHM primarily targets low-frequency, small-amplitude mechanical vibrations or deflections (on the order of centimeters or even millimeters) \cite{GNSS}. For instance, as illustrated in Fig. \ref{scenario}(b)-(g), structural deformations occur in various forms; according to {ISO/TR 4553} \cite{ISO4553_2022}, steel structures permit vertical deflections of 1/500 to 1/150 of their span, single-story buildings allow horizontal displacements up to 1/150 of their height, and the measured vibration frequency must differ from the natural frequency by at least 30\%. These requirements pose two significant challenges for current ISAC systems: \textit{clutter and interference suppression} and \textit{ultra-high sensing precision}.

This paper proposes a reconfigurable intelligent surface (RIS)-aided, cellular-based SHM framework to address these challenges. RIS is introduced to mitigate the impact of clutter and interference on sensing performance. As an emerging technology, RIS can rapidly adjust its phase to tailor the electromagnetic propagation environment, offering a promising software-defined architecture for smart radio environments with reduced cost, size, weight, and power consumption \cite{RISMarco}, \cite{RIS}, \cite{RISliu}. RIS can enhance the energy of selected propagation paths through phase optimization and distinguish its signal from background interference. While current research primarily focuses on optimizing RIS phases for improved communication and sensing coverage \cite{IRSaided}, \cite{RISaided}, new modalities such as active RIS \cite{ARIS}, hybrid RIS \cite{HRIS}, and STAR RIS \cite{starRIS} have emerged. 
The work in \cite{RISLochenk} addresses the problem of the joint 3D localization of a hybrid RIS and a user.
By exploiting the rapidly changing characteristics of RIS, its path can be differentiated from environmental cluster paths and interference, thereby reducing their adverse impact on sensing performance. Moreover, RIS can be seamlessly integrated into building structures, such as walls, billboards, and windows, without compromising aesthetics \cite{cui2017information}. Despite these advantages, the small perturbation introduced by RIS makes it challenging to distinguish from background interference, necessitating further design of RIS phase control. To the best of our knowledge, neither the system model nor the theoretical performance of cellular-based SHM with RIS assistance has been investigated.

Cellular systems inherently provide dense networks that can serve as sensing networks, potentially achieving the ultra-high sensing precision required for SHM through multi-node spatiotemporal cooperation without additional device installations \cite{FIM2}. Moreover, cellular systems offer a unified external reference, with base stations serving as stable anchor points \cite{FIM1}. 
Theoretical analysis frameworks for cooperative localization have been an ongoing research focus.
Cooperative localization approaches and their applications in ultrawideband wireless networks have been reviewed in \cite{NLHenk}. A general framework to determine the fundamental limits of localization accuracy in wideband wireless networks was developed in \cite{EFIM} and later extended to cooperative location-aware networks \cite{NEFIM}. Furthermore, a theoretical foundation based on the equivalent Fisher information matrix (EFIM) for network localization and navigation, including spatiotemporal cooperation, array signal processing, and map exploitation, was presented in \cite{ELIP} to guide the design and operation of practical localization systems. However, these frameworks do not consider the collaboration of cellular networks, as they omit cooperation between base stations and between base stations and users, and do not account for real-world 3D scenarios or the impact of RIS.

  \begin{figure} 
 	\center
 	\includegraphics[width=0.49\textwidth]{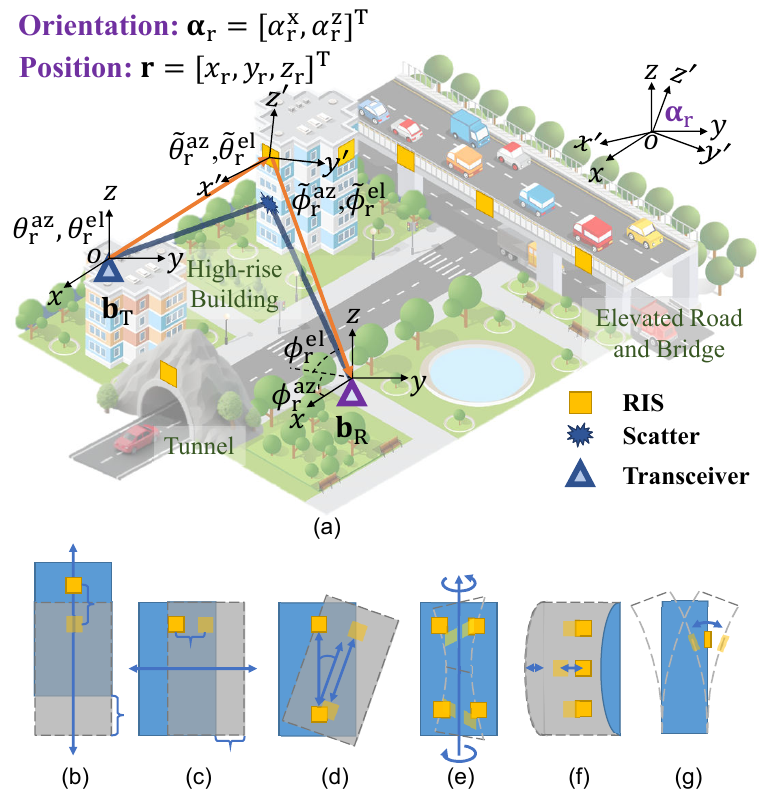} 
 	\caption{(a) 3D scenario of cellular-based SHM with RIS assistance; (b)-(g) Structural deformation illustrations: (b) vertical displacement, (c) horizontal displacement, (d) roll, (e) torsion, (f) bending, and (g) vibration.} 
 	\label{scenario}
 \end{figure}

In this study, RIS are used as reference points for monitoring building structures. By analyzing positional changes of these reference points, the structural state can be determined. 
The cellular-based SHM process comprises three steps: measurement acquisition, localization analysis, and structure state identification.
First, we optimize the RIS phase configuration to obtain measurements and eliminate background noise interference. Then, we analyze the influence of cellular system parameters on reference point localization using Fisher information theory. Additionally, we examine the impact on overall system performance when nodes with self-positioning errors (e.g., mobile users) participate in the collaboration.
Finally, we develop an inference mechanism to analyze the detection probability of structural state identification.
The main contributions of this work are summarized as follows: 
\begin{itemize}  

\item We propose an SHM method that uses RIS as reference points to suppress multipath interference. An analytical framework for structural deformation detection is established using Fisher information and hypothesis testing theory, which accounts for RIS position estimation errors. Both theoretical and numerical analyses confirm the effectiveness of ISAC in detecting millimeter-level deformations.

\item We develop an analytical framework for RIS-aided cellular network collaborative sensing in 3D scenarios based on Fisher information theory. Our analysis reveals how increased observation time, additional receivers (even with self-positioning errors), optimized RIS phases, and refined base station–user collaborative node selection affect the position error bound (PEB), thereby meeting the stringent accuracy requirements of SHM.

\end{itemize}

The remainder of the paper is organized as follows. Section~\ref{sec-system-model} introduces the system model and problem formulation. Section~\ref{sec-Information-Ellipsoid-for-SHM} presents the theoretical analysis framework for SHM. Section~\ref{sec-Reference-Point-Localization} analyzes the sensing accuracy improvements enabled by spatiotemporal collaboration in ISAC networks. Section~\ref{sec-SIM} presents extensive numerical results, and Section~\ref{sec-conclusion} concludes the paper.

\textbf{Notations}: The transpose, conjugate transpose, and inverse
matrices are denoted by $(\ast)^{\rm T}$, $(\ast)^{\rm H}$, and $(\ast)^{-1}$,
respectively. $[\mathbf{a}]_{n}$ represents the $n$-th element of vector $\mathbf{a}$.
The function $\tr(\ast)$ denotes the trace of a matrix,  
$\diag(\ast)$ creates a diagonal matrix, and $\mathbb{E}\{\ast\}$ denotes the statistical expectation.
The operators $\Re\{\ast\}$ and $\Im\{\ast\}$ extract the real and imaginary parts, respectively. The expression $\mathbf{A}\succcurlyeq\mathbf{B}$ indicates that $\mathbf{A}-\mathbf{B}$ is positive semidefinite. The notation 
$|\ast|$ represents the modulus operation for a scalar, and $\| \ast \|_2$ represents the 2-norm of a vector. Finally, $\partial$ denotes the partial derivative.

\section{System Model and Problem Formulation}\label{sec-system-model}
In this section, we present the 3D wireless communication scenario involving a RIS mounted on a potentially deformable structure and formulate the structural deformation detection problem.

\subsection{System Model} 
We consider a 3D wireless communication scenario, as illustrated in Fig. \ref{scenario}(a). A transmitter is equipped with a uniform planar array (UPA) consisting of
$N_{\rm T}=N_{\rm T}^{\rm az} \times N_{\rm T}^{\rm el}$ antenna elements. The central position of the transmitter array is denoted by $\mathbf{b}_{\rm T} \in \mathbb{R}^{3 \times 1}$. Let $\mathbf{p}_{{\rm T}, n} \in \mathbb{R}^{3 \times 1}$ represent the relative position vector of the $n$-th transmit antenna element, with the origin at the array center.
Similarly, a receiver is equipped with a UPA comprising $N_{\rm R}=N_{\rm R}^{\rm az} \times N_{\rm R}^{\rm el}$ antenna elements, and its central position is represented by $\mathbf{b}_{\rm R} \in \mathbb{R}^{3 \times 1}$. 
Let $\mathbf{p}_{{\rm R}, n} \in \mathbb{R}^{3 \times 1}$ denote the relative position vector of the $n$-th receive antenna element. The positions of both the transmitter and receiver are fixed and known.
An RIS, comprising $M$ reflecting elements, is deployed on a building structure, such as a bridge, an elevated road, a high-rise building, or a tunnel. For illustration, we consider a single RIS whose center is given by $\mathbf{r}=[x_{\rm r},y_{\rm r},z_{\rm r}]^{\rm T}$, and whose orientation is characterized by $\boldsymbol{\alpha}_{\rm r}=[ \alpha_{\rm r}^{\rm x},\alpha_{\rm r}^{\rm z}]^{\rm T}$.
Let $\mathbf{p}_{{\rm r}, m} \in \mathbb{R}^{3 \times 1}$ denote the relative position vector of the $m$-th RIS element, with the origin at the array center. We assume that the initial placement and orientation of the RIS is known and denoted by $\mathbf{r}_0$ and $\boldsymbol{\alpha}_{\rm r,0}$, respectively. Note that the UPAs and the RIS elements are arranged along the $y$-$z$ plane. If structural deformation occurs, the RIS position and orientation may change (e.g., rotating to the $y'$-$z'$ plane).

At the $i$-th observation instant, the channel can be written as
\begin{equation} \label{h}
\mathbf{H}^{[i]} = \mathbf{H}_{\rm r}^{[i]}  +  \mathbf{H}_{\rm s}^{[i]}   ,
\end{equation}
where $\mathbf{H}_{\rm r}^{[i]} $ denotes the channel contributions from the RIS, and $\mathbf{H}_{\rm s}^{[i]} $ represents the contribution from scatterers.  Specifically, the RIS channel component is defined as 
\begin{multline} \label{hr} 
	\mathbf{H}_{\rm r}^{[i]}  = g_{\rm r}e^{-j 2\pi f \tau_{r}}\mathbf{a}_{\rm R}( \phi_{\rm r}^{\rm az} ,\phi_{\rm r}^{\rm el} )\mathbf{a}_{\rm r}^{\rm H}( \tilde{\phi}_{\rm r}^{\rm az} ,\tilde{\phi}_{\rm r}^{\rm el} ) \mathbf{\Omega}^{[i]}
	 \\
	\times  \mathbf{a}_{\rm r}( \tilde{\theta}_{\rm r}^{\rm az} ,\tilde{\theta}_{\rm r}^{\rm el} )
	\mathbf{a}_{\rm T}^{\rm H}( \theta_{\rm r}^{\rm az} ,\theta_{\rm r}^{\rm el} ),
\end{multline}
where $g_{\rm r}\in \mathbb{C}$ is the channel coefficient, $\tau_{r}$ is the time of arrival (TOA), and $f$ is the carrier frequency. Here, $\phi_{{\rm r}}^{\rm az}$ and $\phi_{{\rm r}}^{\rm el}$ denote the azimuth and elevation angles of arrival (AOA) at the receiver array, respectively, while $\theta_{{\rm r}}^{\rm az}$ and $\theta_{{\rm r}}^{\rm el}$ denote the azimuth and elevation angles of departure (AOD) from the transmitter array. The corresponding AOD and AOA at the RIS are given by
$\tilde{\phi}_{{\rm r}}^{\rm az}$, $\tilde{\phi}_{{\rm r}}^{\rm el}$, $\tilde{\theta}_{{\rm r}}^{\rm az}$, and $\tilde{\theta}_{{\rm r}}^{\rm el}$, respectively.

To characterize spatial selectivity, each array (Tx, Rx, RIS) is associated with a steering vector $\mathbf{a}_\ast(\cdot)$ ($\ast \in \{ {\rm R}, {\rm T}, {\rm r}\}$), which depends on azimuth and elevation angles. For instance, if the wave vector impinging on the receiver has azimuth angle $\phi_{{\rm r}}^{\rm az}$ and elevation angle $\phi_{{\rm r}}^{\rm el}$, then
\begin{equation} \label{wv}
	 \mathbf{k}(\phi_{{\rm r}}^{\rm az}, \phi_{{\rm r}}^{\rm el}) = \dfrac{2\pi }{\lambda}{\left[\cos\phi_{{\rm r}}^{\rm az} \cos\phi_{{\rm r}}^{\rm el}, \sin\phi_{{\rm r}}^{\rm az}\cos\phi_{{\rm r}}^{\rm el}, \sin\phi_{{\rm r}}^{\rm el}  \right]}^{\rm T},
\end{equation}  
where $\lambda = f/c$ is the wavelength, and $c$ is the speed of light. The $n$-th element  of the receive steering vector is then given by
\begin{equation} \label{sv}
\left[\mathbf{a}_{\rm R}( \phi_{\rm r}^{\rm az} ,\phi_{\rm r}^{\rm el} )\right]_{n} = e^{-j \mathbf{p}^{\rm T}_{{\rm R},n} \mathbf{k}(\phi_{{\rm r}}^{\rm az}, \phi_{{\rm r}}^{\rm el})}, ~n = 1, \ldots, N_{\rm R}.
\end{equation} 
Analogous expressions hold for the transmitter and RIS. 
These steering vectors are left unnormalized to preserve the dependence on the number of antennas for subsequent localization analysis. Additional details on geometric and channel parameters are provided in Appendix~\ref{A0}.

In \eqref{hr}, the matrix $\mathbf{\Omega}^{[i]} $ denotes the RIS phase-shift configuration at time $i$.
Following \cite{RIS}, we approximate $\mathbf{\Omega}^{[i]}$ by an $M \times M$ diagonal matrix:
\begin{equation} \label{Omega}
	\mathbf{\Omega}^{[i]}  = \diag{\left(e^{j\omega^{[i]} _{1}},\ldots,e^{j\omega^{[i]} _{M}} \right)},
\end{equation}
where $\omega^{[i]} _{m}$ represents the phase of the $m$-th RIS element at the $i$-th observation.
For two closely spaced observation instants, ${i-1}$ and $i$, it is reasonable to assume that the surrounding scattering channel remains essentially unchanged:
\begin{equation} \label{H_i-1=H_i}
\mathbf{H}_{\rm s}^{[i]} = \mathbf{H}_{\rm s}^{[i-1]}. 
\end{equation}
Therefore, if the phases of the RIS elements vary over time, the RIS can serve as a dynamic reference point within the radio propagation environment.

From \eqref{h}, the received signal at time $i$ is given by
\begin{equation} \label{y}
	\mathbf{s}^{[i]} = \left(\mathbf{H}_{\rm r}^{[i]}  +  \mathbf{H}_{\rm s}^{[i]} \right) \mathbf{f} \ x +\mathbf{w}^{[i]} ,
\end{equation} 
where $\mathbf{f}$ represents the beamforming vector used for sensing, and $x$ denotes the transmitted pilot signal. We assume that the pilot has a constant envelope, i.e., $ |x|^2 = P_{\rm T} $, where $P_{\rm T}$ is the transmit power.
The vector $\mathbf{w}^{[i]} $ represents additive white Gaussian noise,  with each element distributed as $\mathcal{N}(0,\sigma^2/2)$.
To isolate the information from the dynamic reference point, i.e., the RIS, we subtract the received signals at adjacent observation instants, thereby eliminating the influence of the background propagation environment. Defining 
\begin{equation}
\mathbf{s}_{i,{i-1}} = \mathbf{s}^{[i]} - \mathbf{s}^{[i-1]},
\end{equation}
and 
\begin{equation}
\mathbf{w}_{i,i-1} = \mathbf{w}^{[i]} - \mathbf{w}^{[i-1]},
\end{equation}
and applying \eqref{H_i-1=H_i}, we obtain  
\begin{equation} \label{y-s}
\mathbf{s}_{i,{i-1}}  = \big(\mathbf{H}_{\rm r}^{[i]}-\mathbf{H}_{\rm r}^{[i-1]}\big)\mathbf{f}\ x +\mathbf{w}_{i,i-1},
\end{equation}
where the elements of $\mathbf{w}_{i,i-1}$ follow the normal distribution $\mathcal{N}(0, \sigma^2)$.
For SHM, our goal is to estimate the RIS position from $\mathbf{s}_{i,{i-1}}$ and thereby detect whether any structural deformation has occurred. Note that in this study, we consider far-field scenarios, where the RIS is treated as a point target for the sake of simplification.

\subsection{Problem Formulation}
Because structural deformation detection critically depends on precise RIS localization, we employ a Fisher information-based analysis. Let $\mathbf{y}_{i,{i-1}}$ be the noiseless counterpart of $\mathbf{s}_{i,{i-1}}$. The Fisher information matrix (FIM) for both the channel parameters and the RIS state is 
\begin{equation} \label{FIM}  
	\mathbf{F}(\boldsymbol{*})=\frac{2}{\sigma^2}\Re\left\{ \frac{\partial \mathbf{y}^{\rm T} _{i,{i-1}}}{\partial \boldsymbol{*}}\Big(\frac{\partial \mathbf{y}^{\rm T}_{i,{i-1}}}{\partial \boldsymbol{*}}\Big)^{\rm H}\right\},
\end{equation}
where $\boldsymbol{*} = \{\boldsymbol{\eta}, \mathbf{r}\}$. The vector of channel parameters $\boldsymbol{\eta}$ is defined as
\begin{equation}
    \label{er}
 \boldsymbol{\eta} = \big[\tau_{\rm r}, \phi_{\rm r}^{\rm az}, \phi_{\rm r}^{\rm el}, \tilde{\phi}_{\rm r}^{\rm az}, \tilde{\phi}_{\rm r}^{\rm el}, \tilde{\theta}_{\rm r}^{\rm az}, \tilde{\theta}_{\rm r}^{\rm el}, \theta_{\rm r}^{\rm az}, \theta_{\rm r}^{\rm el}, 
 \Re\{g_{\rm r}\}, \Im\{g_{\rm r}\}\big]^{\rm T}.
\end{equation}  
Given the relationship
\begin{equation} \label{FIM3}
	\frac{\partial \mathbf{y}^{\rm T}_{i,{i-1}}}{\partial \mathbf{r}}=\frac{\partial\boldsymbol{\eta}^{\rm T}}{\partial \mathbf{r}}\frac{\partial \mathbf{y}^{\rm T}_{i,{i-1}}}{\partial \boldsymbol{\eta}},
\end{equation}
we can establish the relationship between the FIM of the RIS state and the FIM of the channel parameters as
\begin{equation} \label{FIM5}
	\mathbf{F}(\mathbf{r})=\mathbf{J} \mathbf{F}(\boldsymbol{\eta})  \mathbf{J}^{\rm H},
\end{equation}
where $\mathbf{J} = {\partial\boldsymbol{\eta}^{\rm T}}/{\partial \mathbf{r}}$ is the Jacobian matrix. 
For any unbiased estimator $\hat{\mathbf{r}}$ of $\mathbf{r}$, the estimation error is bounded by the inverse of the FIM:
\begin{equation} \label{FIM6}
\mathbb{E}\{(\hat{\mathbf{r}}-\mathbf{r})(\hat{\mathbf{r}}-\mathbf{r})^{\rm T}  \} \succcurlyeq \mathbf{F}^{-1}(\mathbf{r}).
\end{equation}  
Consequently, the Cram\'{e}r-Rao lower bound (CRLB) is given by
\begin{equation} \label{CRLB}
	\mbox{CRLB}(\mathbf{r}) =  \Tr\big(\mathbf{F}^{-1}(\mathbf{r}) \big).
\end{equation}
Using the CRLB, we define the squared PEB for the RIS as
\begin{align}\label{SPEB}
	\sigma^2_{\rm p}(\mathbf{r})  &=   \mbox{CRLB}(\mathbf{r})  .
\end{align}

The main objective of this work is to establish a Fisher information-based theoretical framework for structural deformation analysis in cellular networks. Specifically, we aim to localize a dynamic reference point (i.e., the RIS) from received signals and then ascertain whether the structure has deformed based on the RIS’s estimated position. Unlike conventional fixed reference points, such as corner reflectors, the RIS can alter its phase-shift configuration, potentially improving or degrading estimation performance depending on the system design. Two central issues are addressed:
\begin{itemize}  
    \item \textbf{Theoretical Framework for SHM in Cellular Systems with Reference Point Estimation Errors:} 
    Existing research lacks a unified theoretical framework for SHM in cellular systems, providing limited guidance on configuring network parameters, designing RIS phases, and selecting collaborative nodes. These factors are crucial for enhancing SHM accuracy and efficiency. A Fisher information-based approach offers the necessary foundation for jointly optimizing these parameters and assessing the impact of beamforming or RIS-phase mismatches (caused by RIS movements) on monitoring capabilities.
 
    \item \textbf{Spatiotemporal  Collaborative Sensing in 3D Cellular Networks:}  
    Current work has not yet formalized a collaborative sensing analysis for fully 3D scenarios or accounted for the effect of RIS phase control on temporal collaboration. Moreover, multi-node spatial collaboration in cellular networks (involving both base stations and user equipment) often neglects the self-positioning errors introduced by user mobility. The impact of such errors on collaborative sensing performance remains unclear. 
\end{itemize}

\section{Theoretical Analysis Framework for SHM}\label{sec-Information-Ellipsoid-for-SHM} 
\subsection{FIM of Channel Parameters}\label{ris-conf} 
In this subsection, we derive the FIM $\mathbf{F}(\boldsymbol{\eta}) $ for the channel parameter $\boldsymbol{\eta}$ in closed form.  Following \cite{FIM2}, \cite{FIM1}, $\mathbf{F}(\boldsymbol{\eta}) $ can be approximated as a block diagonal matrix. 
Let the scalar $F_{uv} $ denote the element in the matrix $\mathbf{F}(\boldsymbol{\eta}) \in \mathbb{R}^{11 \times 11} $, we have
\begin{align} \label{fim-e}  
F_{uv}  &= \frac{2 }{\sigma^2 }\Re \left\{\Big(\frac{\partial \mathbf{y}^{\rm T}_{i,{i-1}}}{\partial u}\Big) \Big(\frac{\partial \mathbf{y}^{\rm T}_{i,{i-1}}}{\partial v}\Big)^{\rm H}\right\} \notag \\
 &= \frac{2 }{\sigma^2 }\Re  \Big\{
\underbrace{\mu_{uv}  }_{{\rm RX\ factor}} \times \underbrace{ \beta_{uv} }_{ {\rm RIS\ factor}  }\times \underbrace{ \gamma_{uv}  }_{ {\rm TX\ factor}  }
 \Big\},
\end{align}  
where $u$ and $v$ are elements of the vector $\boldsymbol{\eta}$ defined in \eqref{er}.
Based on geometric relationships (see Appendix \ref{A0}), the RX factor ($\mu_{uv}$) involves TOA and AOA parameters at the receiver array,  which depend on the RIS position but not on its orientation. The RIS factor ($\beta_{uv}$) incorporates the observed AOA and AOD at the RIS, which depend on both the RIS position and orientation. Hence, estimating the RIS orientation would require sensing-capable RIS elements \cite{yang2024intelligent}, which lies outside the scope of this study. Finally, the TX factor ($\gamma_{uv}$) involves the AOD parameter from the transmitter array, which depends on the RIS position but not its orientation.


We analyze the RX factor by considering the TOA ($\tau_{\rm r}$) and AOA ($\phi_{\rm r}^{\rm az}, \phi_{\rm r}^{\rm el}$) channel parameters at the receiver array. Notably, the RIS factor ($\beta_{uv}$) and Tx factor ($\gamma_{uv}$) are independent of these parameters and thus remain constant with respect to $\tau_{\rm r}$, $\phi_{\rm r}^{\rm az}$, and $\phi_{\rm r}^{\rm el}$. Consequently, we have
\begin{subequations} \label{beta_gamma_**} 
	\begin{align}
		\beta_{**} & = 	\mathbf{a}_{\rm r}^{\rm H}( \tilde{\theta}_{\rm r}^{\rm az} ,\tilde{\theta}_{\rm r}^{\rm el} )
		\big(\mathbf{\Omega}^{[i]}-\mathbf{\Omega}^{[i-1]}\big)^{\rm H} \mathbf{a}_{\rm r}( \tilde{\phi}_{\rm r}^{\rm az} ,\tilde{\phi}_{\rm r}^{\rm el} )\\ \nonumber
		&~~\times
		\mathbf{a}_{\rm r}^{\rm H}( \tilde{\phi}_{\rm r}^{\rm az} ,\tilde{\phi}_{\rm r}^{\rm el} )
		\big(\mathbf{\Omega}^{[i]}-\mathbf{\Omega}^{[i-1]}\big)
		\mathbf{a}_{\rm r}( \tilde{\theta}_{\rm r}^{\rm az} ,\tilde{\theta}_{\rm r}^{\rm el} ),\\
		\gamma_{**} &=   \mathbf{f}^{\rm H}\mathbf{a}_{\rm T}( \theta_{\rm r}^{\rm az} ,\theta_{\rm r}^{\rm el} ) \mathbf{a}^{\rm H}_{\rm T}( \theta_{\rm r}^{\rm az} ,\theta_{\rm r}^{\rm el} )\mathbf{f} \, |x|^2,
	\end{align}
\end{subequations}  
for $* \in \{\tau_{\rm r}, \phi_{\rm r}^{\rm az}, \phi_{\rm r}^{\rm el} \}$. We then define the per-antenna received signal-to-noise ratio (SNR) as  
\begin{equation} \label{fim-snr}
	{\rm SNR} = \frac{|g_r|^2 \beta_{**} \gamma_{**}}{\sigma^2}.
\end{equation} 

By appropriately designing the RIS phase shifts $\mathbf{\Omega}^{[i]}$ and the beamforming vector $\mathbf{f}$, we can maximize $\beta_{**}$ and $\gamma_{**}$, respectively. Given the known initial positions of both the transceiver and the RIS, each RIS element’s phase can be configured to align with the incoming signal phase and to steer the reflected signal in the desired direction. Let $\varpi^{[i]}_m$ and $\varpi^{[i-1]}_m$ denote the phase of the $m$-th RIS element at time $i$ and ${i-1}$, respectively. Furthermore, let $\bar{\phi}_m$ and $\bar{\theta}_m$ be the phases of the $m$-th element in $\mathbf{a}_{\rm r}( \tilde{\phi}_{\rm r}^{\rm az} ,\tilde{\phi}_{\rm r}^{\rm el} )$ and $\mathbf{a}_{\rm r}( \tilde{\theta}_{\rm r}^{\rm az} ,\tilde{\theta}_{\rm r}^{\rm el} )$, respectively. As shown in Fig. \ref{fig2RIS}, the optimal phases then satisfy 
\begin{equation}
\varpi^{[i]}_m = -(\bar{\phi}_m - \bar{\theta}_m) \text{ and } \varpi^{[i-1]}_m = \pi-(\bar{\phi}_m - \bar{\theta}_m).
\end{equation}
Therefore, the $m$-th diagonal element of the matrix ${\big(\mathbf{\Omega}^{[i]}-\mathbf{\Omega}^{[i-1]}\big)}$ equals to $2e^{-j(\bar{\phi}_m - \bar{\theta}_m)}$. This indicates that the composite phase of the RIS at adjacent time instants is aligned with the transmitter and receiver.
Denote the corresponding RIS configurations at times $i-1$ and $i$ by $\mathbf{\Omega}_{-}$ and $\mathbf{\Omega}_{+}$, respectively.
Then, if the RIS position remains unchanged across these instants, we have $\beta_{**}=4M^4$.
To maintain flexible and optimal control, the RIS alternates its configuration between $\mathbf{\Omega}_{-}$ and $\mathbf{\Omega}_{+}$ at different observation times. Meanwhile, the beamforming vector is designed as ${\mathbf{f}=\mathbf{a}_{\rm T}^{\rm H}( \theta_{\rm r}^{\rm az} ,\theta_{\rm r}^{\rm el} ) }$, thereby ensuring that the transmitted signal is directed toward the RIS. Consequently, if the RIS position remains unchanged, we have $\gamma_{**} = N_{\rm T}^2P_{\rm T}$.  

\begin{figure} 
	\center
	\includegraphics[width=0.42\textwidth]{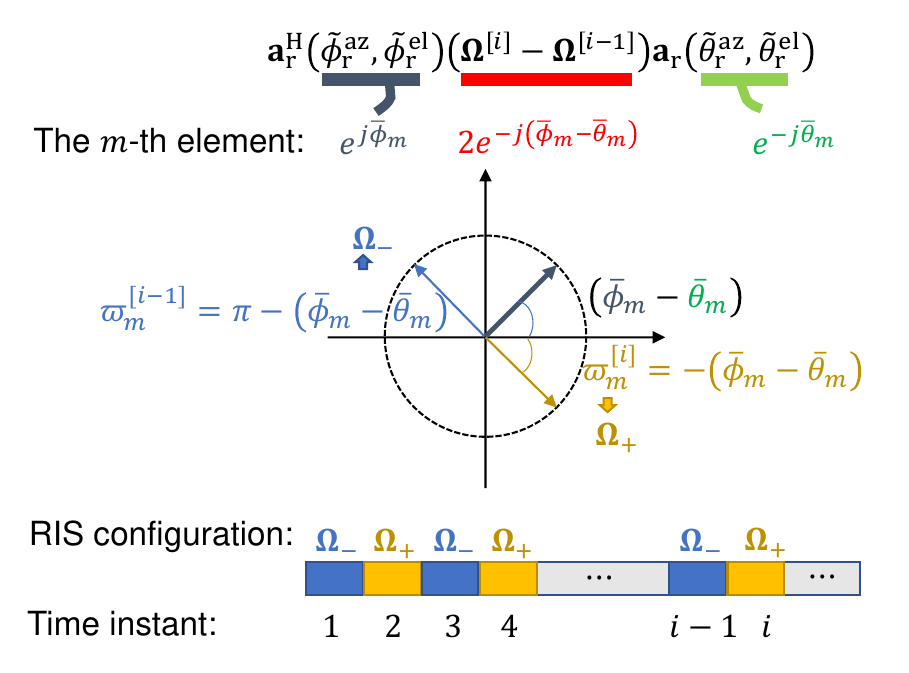} 
	\caption{Illustration of the RIS phase configuration.}
	\label{fig2RIS}
\end{figure}

For the TOA parameter $\tau_{\rm r}$,\footnote{As shown in Fig. \ref{ellips}(a), the delay consists of two parts: the delay from the transmitter to the RIS (denoted as $\tau_{\rm tr}$) and the delay from the RIS to the receiver. A single receiver cannot decouple the two segments of the delay. By using multiple receivers to solve the equations, $\tau_{\rm tr}$ can be estimated, enabling RIS localization.}  we have 
\begin{equation}\label{fim-toa}
\mu_{\tau_{\rm r}\tau_{\rm r}}   = 4\pi^2 f^2 |g_r|^2 N_{\rm R}.
\end{equation}
According to \eqref{fim-e}, the corresponding element in the FIM is
\begin{equation} \label{fim-e1}
F_{\tau_{\rm r} \tau_{\rm r}} = \frac{2 }{\sigma^2 }\Re  \{\mu_{\tau_{\rm r}\tau_{\rm r}} \times \beta_{\tau_{\rm r}\tau_{\rm r}} \times \gamma_{\tau_{\rm r}\tau_{\rm r}} \}.
\end{equation}
Combining \eqref{beta_gamma_**} and \eqref{fim-toa} with \eqref{fim-e1} and simplifying yield 
\begin{equation} \label{fim-e22}
	F_{\tau_{\rm r} \tau_{\rm r}} =8 \pi^2 f^2  N_{\rm R}{\rm SNR},
\end{equation}
indicating that $F_{\tau_{\rm r} \tau_{\rm r}}$ is influenced by the carrier frequency and SNR. 
When considering an orthogonal frequency division multiplexing (OFDM) signal \cite{SLAM}, with $N_{\rm c}$ subcarriers and subcarrier spacing $\Delta f$, and assuming energy normalization across subcarriers, we obtain 
\begin{equation} \label{fim-tt2}
	\mu_{\tau_{\rm r}\tau_{\rm r}}  \approx \frac{ 4\pi^2 (N_{\rm c} \Delta f) ^2 |g_r|^2 N_{\rm R} }{12}.
\end{equation} 
Substituting \eqref{fim-tt2} into \eqref{fim-e1}, we get 
\begin{equation} \label{fim-e2}
	F_{\tau_{\rm r} \tau_{\rm r}} \approx \frac{2 \pi^2  }{3}B^2  N_{\rm R}{\rm SNR},
\end{equation}
where $B=N_{\rm c} \Delta f $ represents the effective bandwidth. 

For the azimuth AOA $\phi_{\rm r}^{\rm az}$, we have
\begin{equation}\label{fim-az2}
\mu_{\phi_{\rm r}^{\rm az}\phi_{\rm r}^{\rm az}}\! \!=\!\! \frac{2\pi^2 d^2 |g_r|^2\!}{3\lambda^2}\!\cos^2\! \!\phi_{\rm r}^{\rm az} \! \cos^2\!\! \phi_{\rm r}^{\rm el} N_{\rm R} ({N_{\rm R}^{\rm az}\!-\!1})({2N_{\rm R}^{\rm az}\!-\!1}),
\end{equation}
leading to 
\begin{equation} \label{fim-az}
	F_{\phi_{\rm r}^{\rm az}\phi_{\rm r}^{\rm az}} =  \frac{ 8 \pi^2 }{3 \lambda^2} G_{\phi_{\rm r}^{\rm az}\phi_{\rm r}^{\rm az}} N_{\rm R}{\rm SNR},
\end{equation}
where  $G_{\phi_{\rm r}^{\rm az}\phi_{\rm r}^{\rm az}} \approx (N^{\rm az}_{\rm R} d \cos\phi_{\rm r}^{\rm az} \cos\phi_{\rm r}^{\rm el} )^2$ represents the effective array aperture, and $d$ is the antenna spacing. Notably, when $d = \lambda/2$, the term $\lambda^2$ cancels out, making $F_{\phi_{\rm r}^{\rm az}\phi_{\rm r}^{\rm az}}$ independent of $\lambda$.

For the elevation AOA $\phi_{\rm r}^{\rm el}$, we have
\begin{multline}\label{fim-el2}
\!\!	\mu_{\phi_{\rm r}^{\rm el} \phi_{\rm r}^{\rm el}} = \frac{2\pi^2\! d^2 |g_r|^2}{3\lambda^2}\! \sin^2\! \phi_{\rm r}^{\rm az}\! \sin^2\! \phi_{\rm r}^{\rm el}\!  N_{\rm R}\! ({N_{\rm R}^{\rm az}\!-\!1})({2N_{\rm R}^{\rm az}\!-\!1})\\
	+\frac{2\pi^2 d^2 |g_r|^2}{3\lambda^2}\! \cos^2\phi_{\rm r}^{\rm el}  N_{\rm R}({N_{\rm R}^{\rm el}\!-\!1})({2N_{\rm R}^{\rm el}\!-\!1})\\
	-\frac{2\pi^2 d^2 |g_r|^2}{\lambda^2} \sin \!\phi_{\rm r}^{\rm az} \sin\!\phi_{\rm r}^{\rm el}\cos\!\phi_{\rm r}^{\rm el} N_{\rm R} (N_{\rm R}^{\rm az}\!-\!1)(N_{\rm R}^{\rm el}\!-\!1),
\end{multline}
resulting in
\begin{equation} \label{fim-el}
	F_{\!\phi_{\rm r}^{\rm el}\phi_{\rm r}^{\rm el}} = \frac{ 8 \pi^2 }{3 \lambda^2} G_{\phi_{\rm r}^{\rm el}\phi_{\rm r}^{\rm el}} N_{\rm R}{\rm SNR},
\end{equation}
where $G_{\phi_{\rm r}^{\rm el}\phi_{\rm r}^{\rm el}}\approx (N^{\rm az}_{\rm R} d \sin\phi_{\rm r}^{\rm az} \sin\phi_{\rm r}^{\rm el} )^2  + (N^{\rm el}_{\rm R} d \cos\phi_{\rm r}^{\rm el} )^2-\frac{3}{2}N_{\rm R}d^2\sin \!\phi_{\rm r}^{\rm az} \sin\!\phi_{\rm r}^{\rm el}\cos\!\phi_{\rm r}^{\rm el}$ represents the effective array aperture.
Notably, $F_{\phi_{\rm r}^{\rm az}\phi_{\rm r}^{\rm az}}$ and $F_{\phi_{\rm r}^{\rm el}\phi_{\rm r}^{\rm el}}$ are primarily influenced by: (i) the number of receiving antennas, (ii) the effective array aperture, and (iii) the SNR.

For the mutual spatial information between $\phi_{\rm r}^{\rm az}$ and $\phi_{\rm r}^{\rm el}$, we have
\begin{multline}\label{fim-azel2}
	\mu_{\phi_{\rm r}^{\rm az} \phi_{\rm r}^{\rm el}} =
	\frac{\pi^2 d^2 |g_r|^2}{\lambda^2}\! \cos \!\phi_{\rm r}^{\rm az} \!\cos^2\!\phi_{\rm r}^{\rm el} N_{\rm R} (N_{\rm R}^{\rm az}\!-\!1)(N_{\rm R}^{\rm el}\!-\!1)\\
	\!\!-\frac{2\pi^2 d^2 |g_r|^2}{3\lambda^2}\!\cos\!\phi_{\rm r}^{\rm az}\! \cos\!\phi_{\rm r}^{\rm el} \! \sin\!\phi_{\rm r}^{\rm az} \!\sin\!\phi_{\rm r}^{\rm el} \! N_{\rm R}\! (\!{N_{\rm R}^{\rm az}\!\!-\!\!1}\!)\!(\!{2N_{\rm R}^{\rm az}\!\!-\!\!1}\!),
\end{multline}  
which leads to the following expression for the FIM element
\begin{equation} \label{fim-elaz}
	F_{\phi_{\rm r}^{\rm az}\phi_{\rm r}^{\rm el}}   
	 = \frac{ 8 \pi^2 }{3 \lambda^2} G_{\phi_{\rm r}^{\rm az}\phi_{\rm r}^{\rm el}} N_{\rm R}{\rm SNR},
\end{equation}
where ${G_{\phi_{\rm r}^{\rm az}\phi_{\rm r}^{\rm el}}\approx \frac{3}{4} N_{\rm R}d^2\cos\!\phi_{\rm r}^{\rm az} \cos^2\!\phi_{\rm r}^{\rm el}\! -\! \cos\!\phi_{\rm r}^{\rm az} \cos\!\phi_{\rm r}^{\rm el}\sin\!\phi_{\rm r}^{\rm az}}$ ${\sin\!\phi_{\rm r}^{\rm el}(N^{\rm az}_{\rm R}d)^2}$.
Notably, we have $F_{\phi_{\rm r}^{\rm el}\phi_{\rm r}^{\rm az}} =	F_{\phi_{\rm r}^{\rm az}\phi_{\rm r}^{\rm el}}$.
 
The detailed derivation of the other elements in $\mathbf{F}(\boldsymbol{\eta})$ is omitted due to space constraints. However, the specific derivation process follows the same approach as those of $F_{\tau_{\rm r} \tau_{\rm r}}$, $F_{\phi_{\rm r}^{\rm az}\phi_{\rm r}^{\rm az}}$,  $F_{\phi_{\rm r}^{\rm el}\phi_{\rm r}^{\rm el}}$,	$F_{\phi_{\rm r}^{\rm az}\phi_{\rm r}^{\rm el}}$ and $F_{\!\phi_{\rm r}^{\rm el}\phi_{\rm r}^{\rm az}}$.
From the analysis above, it can be concluded that several factors influence the FIM of the channel parameters, including:
\begin{itemize}
    \item The effective bandwidth $B$.
    \item The number of antennas on the transceivers.      
    \item The relative positions of the transceivers with respect to the RIS.
    \item The phase configuration of the RIS.
    \item The beamforming strategies employed at the transceivers.
\end{itemize}
By leveraging the initial position and orientation of the RIS, the optimal beamforming and RIS phase configuration can be determined. However, as structural deformations occur, changes in the position and orientation of the RIS may introduce mismatches in both beamforming and phase configurations. This study also theoretically analyzes the impact of such mismatches on the accuracy of structural state identification.

\begin{figure*}[htbp]
	\centering
	\includegraphics[width=0.90\textwidth]{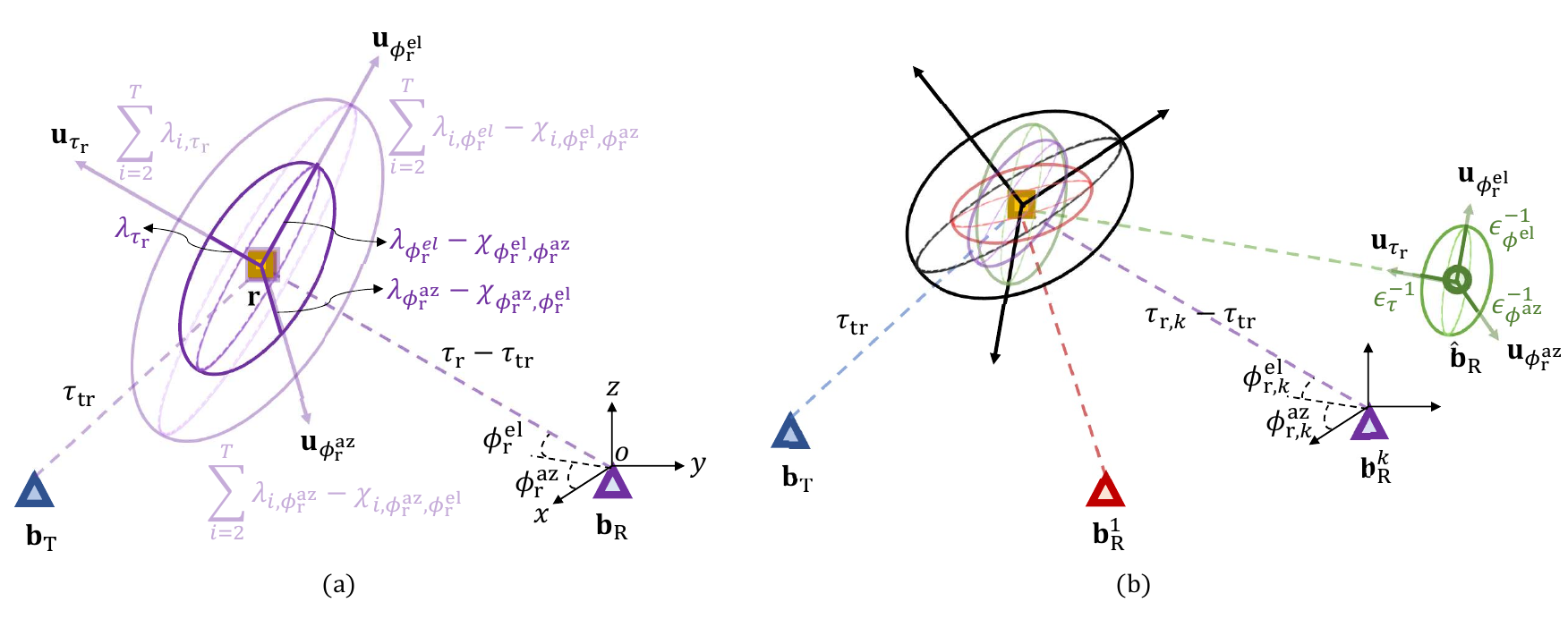}
	\caption{Illustration of information accumulation using information ellipsoids. (a) The information ellipsoids is composed of range information, azimuth angle information, and elevation angle information. Temporal cooperation increases the ellipsoids along three orthogonal directions. (b) Multi-node collaboration with receivers having perfectly known positions and receivers with self-positioning errors.}
	\label{ellips}
\end{figure*}

\subsection{Information Ellipsoid for 3D Localization}
In practical applications, only a subset of $\boldsymbol{\eta}$  is of interest. Let us define the following parameter subsets:
Effective position part $\boldsymbol{\eta}_{\rm pos} = [\tau_{\rm r},\phi_{\rm r}^{\rm az},\phi_{\rm r}^{\rm el}]^{\rm T}$, effective orientation part $\boldsymbol{\eta}_{\rm ori} = [\tilde{\phi}_{\rm r}^{\rm az},\tilde{\phi}_{\rm r}^{\rm el},\tilde{\theta}_{\rm r}^{\rm az},\tilde{\theta}_{\rm r}^{\rm el}]^{\rm T}$, and redundant part $\boldsymbol{\eta}_{\rm red} = [\theta_{\rm r}^{\rm az},\theta_{\rm r}^{\rm el},{\Re}\{g_{\rm r}\},{\Im}\{g_{\rm r}\}]^{\rm T}$.
We define the Jacobian matrix as
 $\mathbf{J}_{\rm A}  = {\partial\boldsymbol{\eta}_{\rm pos}^{\rm T}}/{\partial \mathbf{r}}$, and the submatrix of the FIM as 
\begin{align} \label{D1FIMA}
	\mathbf{F}_{\rm A} &=\left[\begin{array}{ccc}
		F_{\tau_{\rm r} \tau_{\rm r}}  & 0 & 0 \\
		0 & F_{\phi_{\rm r}^{\rm az}\phi_{\rm r}^{\rm az}} & F_{\phi_{\rm r}^{\rm el}\phi_{\rm r}^{\rm az}} \\
		0 & F_{\phi_{\rm r}^{\rm az}\phi_{\rm r}^{\rm el}} & F_{\phi_{\rm r}^{\rm el}\phi_{\rm r}^{\rm el}}
	\end{array}\right].
\end{align}
Thus, the partitioned structure of $\mathbf{F}(\boldsymbol{\eta})$ and $\mathbf{J}$ is expressed as
\begin{equation} 
    \mathbf{F}(\boldsymbol{\eta}) =\left[\begin{array}{cc}
	\mathbf{F}_{\rm A} & \mathbf{F}_{\rm B}\\
	\mathbf{F}_{\rm B}^{\rm T} & \mathbf{F}_{\rm C}
\end{array}\right] 
\text{~and~}  
	\mathbf{J} =\left[ \mathbf{J}_{\rm A} ,\mathbf{J}_{\rm C} \right],
\end{equation}
where $\mathbf{F}_{\rm B}=\mathbf{0}$ in accordance with \cite{FIM2}, \cite{FIM1}. Consequently, the FIM can be rewritten as
\begin{equation} \label{D1FIM3}
\mathbf{J} \mathbf{F}(\boldsymbol{\eta})  \mathbf{J}^{\rm H} = \mathbf{J}_{\rm A} 	\mathbf{F}_{\rm A}  \mathbf{J}_{\rm A}^{\rm H} + \mathbf{J}_{\rm C} 	\mathbf{F}_{\rm C}  \mathbf{J}_{\rm C}^{\rm H}.
\end{equation}
When the orientation and redundant parameters, $\boldsymbol{\eta}_{\rm ori}$ and $\boldsymbol{\eta}_{\rm red}$, are not estimated, their contributions to the localization FIM are zero. According to \eqref{FIM5}, we have
\begin{equation} \label{D1FIM4}
\mathbf{F}(\mathbf{r})= \mathbf{J}_{\rm A} 	\mathbf{F}_{\rm A}  \mathbf{J}_{\rm A}^{\rm H}.
\end{equation}
The element-wise expression of $\mathbf{J}_{\rm A} $ is provided in Appendix \ref{B}.

\begin{theorem}\label{T0}
	The EFIM for 3D localization is given by 
\begin{align}\label{THE1}
	\mathbf{F}_{\rm e}(\mathbf{r})& = \lambda_{\tau_{\rm r}} \mathbf{u}_{\tau_{\rm r}}\mathbf{u}_{\tau_{\rm r}}^{\rm H} + (\lambda_{\phi_{\rm r}^{\rm az}} - \chi_{\phi_{\rm r}^{\rm az}\phi_{\rm r}^{\rm el}}) \mathbf{u}_{\phi_{\rm r}^{\rm az}}\mathbf{u}_{\phi_{\rm r}^{\rm az}}^{\rm H} \nonumber\\
	&~~~+ (\lambda_{\phi_{\rm r}^{\rm el}}-\chi_{\phi_{\rm r}^{\rm el}\phi_{\rm r}^{\rm az}} )\mathbf{u}_{\phi_{\rm r}^{\rm el}}\mathbf{u}_{\phi_{\rm r}^{\rm el}}^{\rm H},
\end{align}
	where $\mathbf{u}_{\tau_{\rm r}}$, $\mathbf{u}_{\phi_{\rm r}^{\rm az}}$, and $\mathbf{u}_{\phi_{\rm r}^{\rm el}}$ are unit direction vectors for range, azimuth AOA, and elevation AOA information, respectively. Their expressions are given by 
	\begin{subequations}\label{unitv}
		\begin{align}
			\mathbf{u}_{\tau_{\rm r}} &=  [ \cos\phi_{\rm r}^{\rm az}\cos\phi_{\rm r}^{\rm el},\sin\phi_{\rm r}^{\rm az}\cos\phi_{\rm r}^{\rm el},\sin\phi_{\rm r}^{\rm el} ]^{\rm T},   \\
			\mathbf{u}_{\phi_{\rm r}^{\rm az}} &=  [-\sin\phi_{\rm r}^{\rm az},\cos\phi_{\rm r}^{\rm az}, 0]^{\rm T}, \\
			\mathbf{u}_{\phi_{\rm r}^{\rm el}} &= [-\cos\phi_{\rm r}^{\rm az}\sin\phi_{\rm r}^{\rm el},-\sin\phi_{\rm r}^{\rm az}\sin\phi_{\rm r}^{\rm el},\cos\phi_{\rm r}^{\rm el} ]^{\rm T}.
		\end{align}
	\end{subequations}
	Furthermore, the terms $ \lambda_{\tau_{\rm r}} $, $(\lambda_{\phi_{\rm r}^{\rm az}}-\chi_{\phi_{\rm r}^{\rm az}\phi_{\rm r}^{\rm el}})$, and $(\lambda_{\phi_{\rm r}^{\rm el}}-\chi_{\phi_{\rm r}^{\rm el}\phi_{\rm r}^{\rm az}})$ represent the information intensities for range, azimuth AOA, and elevation AOA, respectively. The terms $\chi_{\phi_{\rm r}^{\rm az}\phi_{\rm r}^{\rm el}}$ and $\chi_{\phi_{\rm r}^{\rm el}\phi_{\rm r}^{\rm az}}$ denote mutual spatial information intensity. These are expressed as follows:
	\begin{subequations}\label{intensity}
		\begin{align}
			\lambda_{\tau_{\rm r}} &= \frac{ 2\pi^2 B^2 }{3c^2}  N_{\rm R}{\rm SNR}, \label{intensityT} \\
			\lambda_{\phi_{\rm r}^{\rm az}} &= \frac{ 8 \pi^2 }{3 \lambda^2} \frac{G_{\phi_{\rm r}^{\rm az}\phi_{\rm r}^{\rm az}}}{\| \mathbf{r}-\mathbf{b}_{\rm R} \|_2^2\cos^2 \phi_{\rm r}^{\rm el}}   N_{\rm R}{\rm SNR}, \label{intensityAZ}\\
			\lambda_{\phi_{\rm r}^{\rm el}} &=\frac{ 8 \pi^2 }{3 \lambda^2 } \frac{G_{\phi_{\rm r}^{\rm el}\phi_{\rm r}^{\rm el}} } {\| \mathbf{r}-\mathbf{b}_{\rm R} \|_2^2} N_{\rm R}{\rm SNR}, \label{intensityEL} \\
			\chi_{\phi_{\rm r}^{\rm az}\phi_{\rm r}^{\rm el}} &= \frac{8 \pi^2}{3 \lambda^2 }\frac{G_{\phi_{\rm r}^{\rm az}\phi_{\rm r}^{\rm el}}G_{\phi_{\rm r}^{\rm el}\phi_{\rm r}^{\rm az}}}{G_{\phi_{\rm r}^{\rm el}\phi_{\rm r}^{\rm el}}\| \mathbf{r}-\mathbf{b}_{\rm R} \|_2^2\cos^2 \phi_{\rm r}^{\rm el}} N_{\rm R}{\rm SNR}, \\
			\chi_{\phi_{\rm r}^{\rm el}\phi_{\rm r}^{\rm az}} &= \frac{8 \pi^2}{3 \lambda^2 }\frac{G_{\phi_{\rm r}^{\rm az}\phi_{\rm r}^{\rm el}}G_{\phi_{\rm r}^{\rm el}\phi_{\rm r}^{\rm az}}}{G_{\phi_{\rm r}^{\rm az}\phi_{\rm r}^{\rm az}}\| \mathbf{r}-\mathbf{b}_{\rm R} \|_2^2} N_{\rm R}{\rm SNR}.\label{intensityMU}
		\end{align}
	\end{subequations}
\end{theorem}
The proof is provided in Appendix \ref{B}.

As illustrated in Fig. \ref{ellips}(a), the vectors $\mathbf{u}_{\tau_{\rm r}}$, $\mathbf{u}_{\phi_{\rm r}^{\rm az}}$ and $\mathbf{u}_{\phi_{\rm r}^{\rm el}}$ are mutually orthogonal. Specifically, the vector $\mathbf{u}_{\tau_{\rm r}}$ points along the line connecting the reference point $\mathbf{b}_{\rm R}$ to $\mathbf{r}$, while $\mathbf{u}_{\phi_{\rm r}^{\rm az}}$ and $\mathbf{u}_{\phi_{\rm r}^{\rm el}}$ span a plane perpendicular to $\mathbf{u}_{\tau_{\rm r}}$.
An ellipsoid characterized by these axes, with corresponding ``radii'' (information intensities)  $ \lambda_{\tau_{\rm r}} $, $(\lambda_{\phi_{\rm r}^{\rm az}}-\chi_{\phi_{\rm r}^{\rm az}\phi_{\rm r}^{\rm el}})$, and $(\lambda_{\phi_{\rm r}^{\rm el}}-\chi_{\phi_{\rm r}^{\rm el}\phi_{\rm r}^{\rm az}})$, is referred to as the \emph{information ellipsoid}.
Several observations regarding the information ellipsoid are noted: 
\begin{itemize}
 \item It can be shown that the values of $ \lambda_{\tau_{\rm r}} $, $\lambda_{\phi_{\rm r}^{\rm az}}$, $\chi_{\phi_{\rm r}^{\rm az}\phi_{\rm r}^{\rm el}}$, $\lambda_{\phi_{\rm r}^{\rm el}}$, and $\chi_{\phi_{\rm r}^{\rm el}\phi_{\rm r}^{\rm az}}$ are non-negative. Moreover, the non-negativity of $(\lambda_{\phi_{\rm r}^{\rm az}}-\chi_{\phi_{\rm r}^{\rm az}\phi_{\rm r}^{\rm el}})$, and $(\lambda_{\phi_{\rm r}^{\rm el}}-\chi_{\phi_{\rm r}^{\rm el}\phi_{\rm r}^{\rm az}})$ can be rigorously established through detailed calculations involving equations \eqref{fim-az2}, \eqref{fim-el2}, and \eqref{fim-azel2}. Specifically, by computing the product of \eqref{fim-az2} and \eqref{fim-el2}, and subtracting the square of \eqref{fim-azel2}, it can be shown that these expressions are inherently non-negative (Appendix \ref{A3}).
 	
 \item In range estimation, the information intensity is proportional to the square of the signal bandwidth, as indicated by \eqref{intensityT}. This relationship is typically exploited by observing the phase differences among the various subcarriers in OFDM systems. In contrast, for absolute carrier phase estimation, $F_{\tau_{\rm r} \tau_{\rm r}}$ is given by \eqref{fim-e22}. Consequently, the range information intensity (RII) $ \lambda_{\tau_{\rm r}} $ becomes ${ 8\pi^2 f^2}  {\rm SNR}/{c^2}$, which is proportional to the square of the carrier frequency. Since the carrier frequency is generally much higher than the signal bandwidth, absolute phase estimation facilitates more precise range measurements, albeit at the expense of requiring specialized estimation techniques.
 
 \item In angle estimation, \eqref{intensityAZ} to \eqref{intensityMU} demonstrate that the angle information intensity (AII) is determined by the effective aperture of the antenna array. When the distance between the receiving antenna and the target is constant, increasing the number of antennas or enlarging their spacing enhances the effective aperture. Conversely, if the number of antennas and their spacing remain fixed, the effective aperture decreases as the target distance increases or as the angular offset grows. Therefore, angle measurement is generally more suitable for short-range sensing, whereas the RII  is unaffected by the sensing distance.

 \item The TOA and AOAs at the receiver provide sufficient information to estimate the RIS position. If the AOA and AOD at the RIS are also available, it is possible to estimate the RIS orientation using, e.g., \eqref{GR3} and \eqref{GR4}. When analyzing \eqref{D1FIM4}, one could in principle include the derivatives of the TOA and AOA with respect to the orientation in $\mathbf{J}_{\rm A}$. However, since these derivatives are zero, the final result is unaffected. For simplicity, orientation is therefore not considered in the present analysis. The Fisher information related to RIS orientation can be examined using the same methodology, which is an interesting direction for future research.
\end{itemize}

\subsection{Structure State Identification}
To determine whether the RIS position has changed, we compare the estimated RIS position with its initial value, $\mathbf{r}_0$. This problem can be formulated as a classical hypothesis testing problem: 
\begin{align} \label{HT}
	\mathcal{H}_0 & : \mathbf{r} = \mathbf{r}_{0},\\
	\mathcal{H}_1 & : \mathbf{r} \ne \mathbf{r}_{0}.
\end{align}
Under the null hypothesis 
$\mathcal{H}_0$, the RIS is assumed to remain at its initial position, while the alternative hypothesis $\mathcal{H}_1$ indicates that a positional change has occurred.
To account for measurement uncertainties and distinguish them from genuine positional changes, we employ the Wald test. Specifically, the test decides in favor of $\mathcal{H}_1$ if the following condition is satisfied:
\begin{equation} \label{wt} 
	T_W = (\hat{\mathbf{r}}-\mathbf{r}_{0})^{\rm T}\mathbf{F}(\hat{\mathbf{r}})(\hat{\mathbf{r}}-\mathbf{r}_{0}) > \kappa,
\end{equation}
where $\kappa$ is an empirically determined threshold,\footnote{For example, a displacement change of $50$ mm is set as the threshold ($\kappa=50$ mm) for the ``Damaged" state in ordinary buildings and structures, as stated in \cite{SHM}. Generally, the threshold $\kappa$ can be set in accordance with industry standards for structural engineering \cite{ISO4553_2022}. } and $\hat{\mathbf{r}}$ represents the maximum likelihood estimate of $\mathbf{r}$ under $\mathcal{H}_1$. 

In practice, beyond simply detecting a change in the reference point position, the structure state is defined in terms of the serviceability limit state as outlined in \cite{ISO4553_2022}. We assume that the structure state $S$  takes values from a discrete set $\mathbb{S}=\{S_1={\sf U},\ S_2={\sf D} \}$, where ``${\sf U}$'' represent ``Unchanged'' and ``${\sf D}$'' represent ``Damaged'' \cite{Detect}. Based on the Wald test in \eqref{wt}, we define the weighted displacement as
\begin{equation}
	d(\hat{\mathbf{r}}) = (\hat{\mathbf{r}}-\mathbf{r}_{0})^{\rm T}\mathbf{F}(\hat{\mathbf{r}})(\hat{\mathbf{r}}-\mathbf{r}_{0}).
\end{equation}	
The defined weighted displacement not only considers the discrepancy between the measured value and the true value $(\hat{\mathbf{r}}-\mathbf{r}_{0})$, but also takes the reliability of this measurement $\mathbf{F}(\hat{\mathbf{r}})$ into account.
When $d(\hat{\mathbf{r}})$ exceeds the serviceability limit threshold $\kappa$ specified in \cite{ISO4553_2022}, damage is inferred; thus, $\mathcal{H}_0$ corresponds to the state ${\sf U}$ and $\mathcal{H}_1$ to state ${\sf D}$.

The identification of the structure state can be approached as an inference problem. Using Bayes' theorem, the posterior probability is calculated as 
\begin{equation} \label{posterior}
	{\rm P}\big(S_l | d(\hat{\mathbf{r}})\big) =  \frac{f\big(d(\hat{\mathbf{r}})|S_l\big){\rm P}\big(S_l\big)}{f\big(d(\hat{\mathbf{r}})\big)},
\end{equation}
where ${\rm P}(S_l)$ is the prior probability, determined by the structure's attributes, and $f\big( d(\hat{\mathbf{r}})|S_l \big)$ is the likelihood function, which is derived from the parameter estimation implementation. The denominator, $f\big(d(\hat{\mathbf{r}})\big)$, is computed as $  \sum_{l=1}^{2} f\big(d(\hat{\mathbf{r}})|S_l\big){\rm P}\big(S_l\big)$.

\begin{figure}
	\centering
	\includegraphics[width=0.34\textwidth]{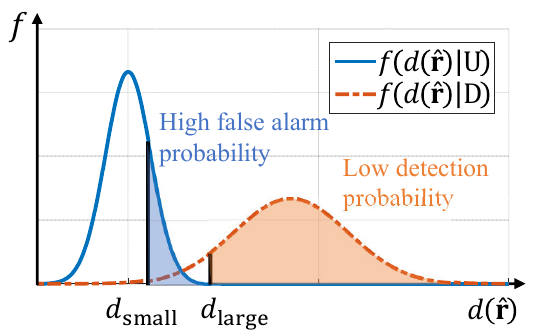}
	\caption{The tradeoff between false alarm probability and detection probability in SHM.}
	\label{detect}
\end{figure}

The likelihood function is designed based on the derivation of the PEB.
Initially, we encode knowledge of the ${\sf U}$ state into the likelihood function for $\mathcal{H}_0$, which provides the probability density $f\big(d(\hat{\mathbf{r}})|{\sf U}\big)$ of $d(\hat{\mathbf{r}})$.
We assume that $f\big(d(\hat{\mathbf{r}})|{\sf U}\big)$ follows a Gaussian distribution with zero mean and variance $\sigma^2_{\rm p}(\mathbf{r}_0)$.
The likelihood function under   $\mathcal{H}_1$, $f(d(\hat{\mathbf{r}})|{\sf D})$, is modeled as a Gaussian distribution with a mean of $d(\hat{\mathbf{r}})$ and variance $\sigma^2_{\rm p}(\hat{\mathbf{r}})$.

The PEB governs the shapes of the likelihood functions $f\big(d(\hat{\mathbf{r}})|{\sf U}\big)$ and $f\big(d(\hat{\mathbf{r}})|{\sf D}\big)$, thereby impacting the detection probability of the damaged state, as illustrated in Fig. \ref{detect}. 
Our observations are twofold:
\begin{itemize}
	\item {\bf PEB:} When the threshold $\kappa$ remains constant, a smaller PEB results in a sharper distribution function, reducing the overlap between the two functions. This reduction in overlap lowers the false alarm probability and increases the detection probability.
	
	\item {\bf Threshold Selection Trade-off:} As the deformation detection requirements become more stringent (i.e., with a lower $\kappa$), the overlap between the two distributions increases. This introduces a trade-off between detection probability and false alarm probability. A lower threshold ($d_{\rm small}$) increases the false alarm rate, whereas a higher threshold ($d_{\rm large}$) reduces the detection probability.
\end{itemize}
Consequently, our objective is to minimize the PEB through the optimal configuration of the cellular system to achieve the required structure state detection performance.
In addition, by combining displacement monitoring from multiple reference points, more complex types of deformation, such as roll, torsion, and bending [Fig. \ref{scenario}(b)-(g)], can be identified. Furthermore, continuous monitoring can incorporate Bayesian principles for multistage decision-making. These are all worth exploring in future research.

\section{Spatiotemporal Cooperation in ISAC Networks}\label{sec-Reference-Point-Localization}
In this section, we explore the opportunities for spatiotemporal cooperation in ISAC networks. Specifically, we analyze the 3D localization performance of the reference point (i.e., the RIS) under three distinct scenarios, which are discussed in the following subsections: (A) increasing the number of observation time instants, (B) increasing the number of cooperative receivers, and (C) increasing the number of cooperative receivers with self-positioning errors.

\subsection{Increasing Observation Time Instants}

Since structural changes often occur gradually, the idle resources of the transceivers during off-peak periods can be effectively utilized. This study aims to investigate whether extending the observation time instants can enhance the sensing performance.

For $T$ observation time instants, the accumulated measurements are expressed as
\begin{equation} \label{y--1}
	\bar{\mathbf{y}}_{[2:T]}= {\left[\mathbf{y}_{2,1}^{\rm T},\ldots, \mathbf{y}_{T,{T-1}}^{\rm T}\right]}^{\rm T},
\end{equation}
where $\bar{\mathbf{y}}_{[2:T]} \in \mathbb{C}^{N_{\rm R}(T-1)\times1}$.
The FIM for localization over $T$ periods is given by
\begin{align} \label{FIM21}
\mathbf{F}_{[2:T]}(\mathbf{r})&=\frac{2 }{\sigma^2 }\Re\left\{  \frac{\partial \bar{\mathbf{y}}^{\rm T}_{[2:T]}}{\partial \mathbf{r}}\left(\frac{\partial \bar{\mathbf{y}}^{\rm T}_{[2:T]}}{\partial \mathbf{r}}\right)^{\rm H}\right\} \notag \\
&= \frac{2 }{\sigma^2 } \mathbf{J}  \Re\left\{ \sum_{i=2}^{T}\frac{\partial \mathbf{y}^{\rm T}_{i,{i-1}}}{\partial \boldsymbol{\eta}} \left(\frac{\partial \mathbf{y}^{\rm T}_{i,{i-1}}}{\partial \boldsymbol{\eta}}\right)^{\rm H} \right\} \mathbf{J}^{\rm H}.
\end{align}
According to Theorem \ref{T0}, the following corollary holds.

\begin{corollary}\label{CO1}
The EFIM for increased observation time instants is given by	
\begin{align}\label{EFIM4}  
	\mathbf{F}_{{\rm e }{[2:T]}}(\mathbf{r}) &= \sum_{i=2}^{T}  \lambda_{i,\tau_{{\rm r}}} \mathbf{u}_{\tau_{\rm r}}\mathbf{u}_{\tau_{\rm r}}^{\rm H}   \nonumber\\
	&~~+ \sum_{i=2}^{T} (\lambda_{i,\phi_{{\rm r}}^{\rm az}} -\chi_{i,\phi_{\rm r}^{\rm az}\phi_{\rm r}^{\rm el}})  \mathbf{u}_{\phi_{\rm r}^{\rm az}}\mathbf{u}_{\phi_{\rm r}^{\rm az}}^{\rm H}  \nonumber\\
	&~~~~+ \sum_{i=2}^{T}(\lambda_{i,\phi_{{\rm r}}^{\rm el}}-\chi_{i,\phi_{\rm r}^{\rm el}\phi_{\rm r}^{\rm az}}) \mathbf{u}_{\phi_{\rm r}^{\rm el}}\mathbf{u}_{\phi_{\rm r}^{\rm el}}^{\rm H}.
\end{align}
In \eqref{EFIM4}, the range and angle information at different time instants are aligned along consistent directions.
Moreover, time accumulation yields localization performance gains regardless of the RIS phase configurations. If the RIS phases are optimized,\footnote{The optimal phase configuration alternates between $\mathbf{\Omega}_{-}$ and $\mathbf{\Omega}_{+}$ for consecutive observation periods.} according to \eqref{THE1}, the EFIM simplifies to 
\begin{equation} \label{FIM22}
	\mathbf{F}_{{\rm e }{[2:T]}}(\mathbf{r}) =(T-1)\mathbf{F}_{\rm e }(\mathbf{r}) .
\end{equation}
\end{corollary}

Corollary \ref{CO1} implies that increasing the number of observation time instants expands the information ellipsoid, as shown in Fig. \ref{ellips}(a), thereby effectively reducing localization errors.

\subsection{Increasing the Number of Receivers}\label{mr} 
A key advantage of cellular systems compared to other sensing systems is their dense network coverage.
By utilizing multiple base stations cooperatively, more accurate sensing results can be achieved.

For $K$ cooperative receivers, the accumulated measurements can be expressed as 
\begin{equation} \label{y--2}
	\bar{\mathbf{y}}^{[1:K]}= \left[({\mathbf{y}_{i,i-1}^1})^{\rm T},\ldots, ({\mathbf{y}_{i,i-1}^K})^{\rm T} \right]^{\rm T}.
\end{equation}
The total channel parameters for $K$ receivers are extended as
\begin{equation}
\bar{\boldsymbol{\eta}}=[\boldsymbol{\eta}_1^{\rm T}, \ldots,\boldsymbol{\eta}_K^{\rm T}]^{\rm T},
\end{equation}
where $\boldsymbol{\eta}_k\!=\![\tau_{{\rm r},k},\!\phi_{{\rm r},k}^{\rm az},\!\phi_{{\rm r},k}^{\rm el},\!\tilde{\phi}_{{\rm r},k}^{\rm az},\!\tilde{\phi}_{{\rm r},k}^{\rm el},\!\tilde{\theta}_{{\rm r},k}^{\rm az},\!\tilde{\theta}_{{\rm r},k}^{\rm el},\! \theta_{{\rm r},k}^{\rm az},\!\theta_{{\rm r},k}^{\rm el},\!{\Re}\{g_{{\rm r},k}\}\!,\\ {\Im}\{g_{{\rm r},k}\}]^{\rm T}.$
Consequently, the FIM for $\bar{\boldsymbol{\eta}}$ is given by
\begin{equation} \label{FIMeta1}
	\mathbf{F}^{[1:K]}(\bar{\boldsymbol{\eta}})=\frac{2 }{\sigma^2 }\Re\left\{ \frac{\partial (\bar{\mathbf{y}}^{[1:K]})^{\rm T}}{\partial \bar{\boldsymbol{\eta}} }\left(\frac{\partial (\bar{\mathbf{y}}^{[1:K]})^{\rm T}}{\partial \bar{\boldsymbol{\eta}} }\right)^{\rm H} \right\}.
\end{equation}
Define $	\bar{\mathbf{J}} =\big[\mathbf{J}_1,  \mathbf{J}_2,  \ldots , \mathbf{J}_K\big]$,
where $ \mathbf{J}_k = {\partial\boldsymbol{\eta}_k^{\rm T}}/{\partial \mathbf{r}}$.
Therefore, the resulting FIM is
\begin{equation} \label{FIM66}
\mathbf{F}^{[1:K]}(\mathbf{r})=\bar{\mathbf{J}}	F(\bar{\boldsymbol{\eta}})  \bar{\mathbf{J}} ^{\rm H}.
\end{equation} 
According to Theorem~\ref{T0}, we have the following corollary.
\begin{corollary}\label{C3}
	The EFIM for an increased number of receivers is given by	
\begin{align}\label{EFIM44}
	\mathbf{F}_{\rm e }^{[1:K]}(\mathbf{r}) & = \sum_{k=1}^{K} \lambda_{\tau_{\rm r}}^{k} \mathbf{u}^{k}_{\tau_{\rm r}}\mathbf{u}^{k{\rm H}}_{\tau_{\rm r}} \nonumber\\
	&~~+ \sum_{k=1}^{K} (\lambda_{\phi_{{\rm r}}^{\rm az}}^{k}-\chi_{\phi_{\rm r}^{\rm az}\phi_{\rm r}^{\rm el}}^{k}) \mathbf{u}^{k}_{\phi_{\rm r}^{\rm az}}\mathbf{u}^{k{\rm H}}_{\phi_{\rm r}^{\rm az}} \nonumber\\
	&~~+ \sum_{k=1}^{K} (\lambda^{k}_{\phi_{{\rm r}}^{\rm el}}-\chi_{\phi_{\rm r}^{\rm el}\phi_{\rm r}^{\rm az}}^{k})\mathbf{u}^{k}_{\phi_{\rm r}^{\rm el}}\mathbf{u}^{k{{\rm H}}}_{\phi_{\rm r}^{\rm el}},
\end{align}
where the range and angle information from different receivers generally point in different directions. 
\end{corollary}

\begin{figure*}[htbp]
	\centering
	\includegraphics[width=0.9\textwidth]{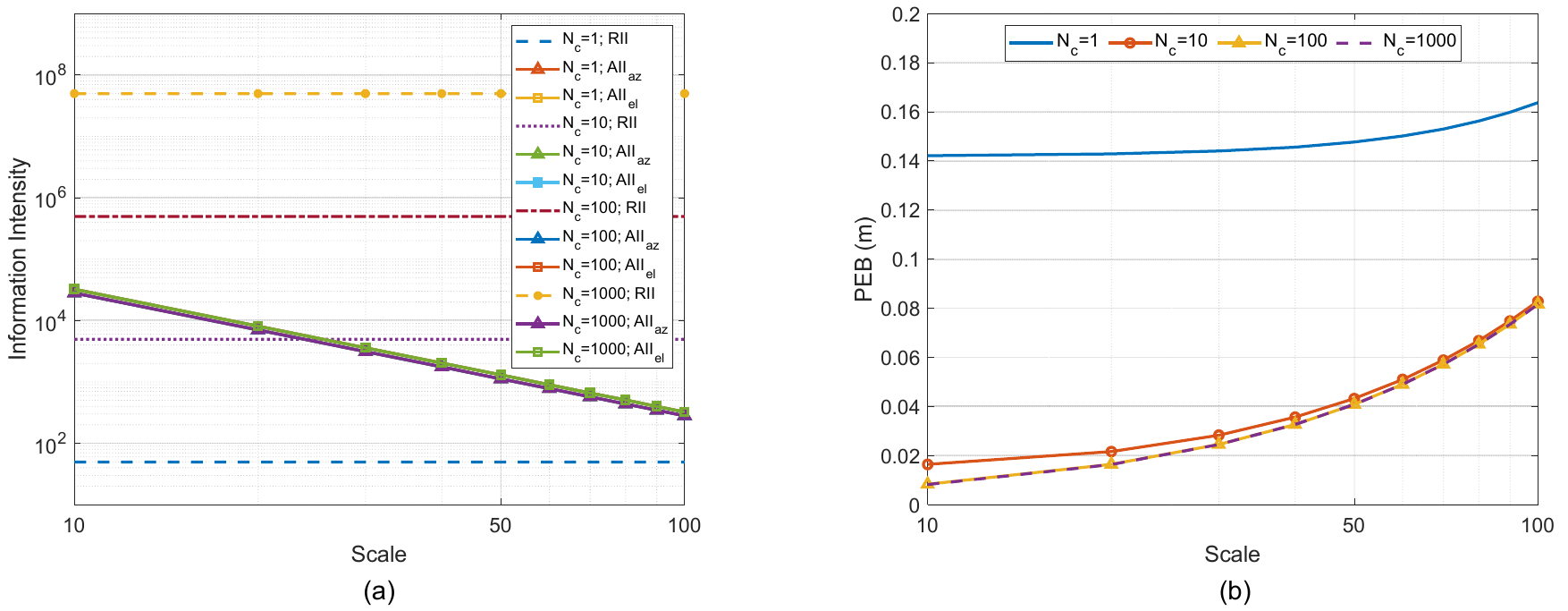}
	\caption{Relationship among scene scale, bandwidth, and (a) RII and AII, and (b) PEB.}
	\label{fig1} 
    \vspace{-0.3cm}
\end{figure*}

Since the directional matrix associated with each receiver depends on its position, as shown in Fig. \ref{ellips}(b), optimizing receiver placement can make the information ellipsoid more isotropic, resembling a uniform sphere, particularly when the number of cooperative receivers is limited. This ensures that the reference point has a uniform localization capability in 3D space.
The claim that the CRLB decreases with an increasing number of receivers is supported by the following corollary.
\begin{corollary}\label{C2}
	Given $\bar{\mathbf{y}}^{[1:K+1]}$, the CRLB satisfies the following inequality:
	\begin{equation} \label{eq:th2} 
		{\rm CRLB}^{[1:K+1]}(\mathbf{r})  \leq {\rm CRLB}^{[1:K]}(\mathbf{r}).
		\end{equation}
\end{corollary}

\begin{proof}
According to \eqref{EFIM44}, we have 
\begin{equation} \label{EFIM35}
	\mathbf{F}_{\rm e }^{[1:K+1]}(\mathbf{r}) \!=\!\!\! \sum_{k=1}^{K+1}\!\! \mathbf{U}^k 	\!\!
	\left[\!\!\begin{array}{ccc}
	\!	\lambda_{\tau_{\rm r}}^{k} \!\!& 0                                  & 0  \\		
		0                          & \!\!\lambda_{\phi_{{\rm r}}^{\rm az}}^{k}\!-\!\chi_{\phi_{\rm r}^{\rm az}\phi_{\rm r}^{\rm el}}^{k}  \!\!  & 0  \\
		0                          & 0                                  &\!\! \lambda^{k}_{\phi_{{\rm r}}^{\rm el}}\!-\!\chi_{\phi_{\rm r}^{\rm el}\phi_{\rm r}^{\rm az}}^{k}\!\!\\
	\end{array}\!\!\right] \!\!\mathbf{U}^{k{\rm H}}, 	
\end{equation}
where $\mathbf{U}^k = [\mathbf{u}_{\tau_{\rm r}}^k,\mathbf{u}_{\phi_{\rm r}^{\rm az}}^k,\mathbf{u}_{\phi_{\rm r}^{\rm el}}^k]$.
Since each summation matrix in \eqref{EFIM35} is positive semi-definite, it follows that
\begin{equation} \label{EFIM45}
	\mathbf{F}_{\rm e }^{[1:K+1]}(\mathbf{r}) \succcurlyeq 	\mathbf{F}_{\rm e }^{[1:K]}(\mathbf{r}). 	
\end{equation}
Therefore, we have
\begin{equation} \label{EFIM65}
	\big[\mathbf{F}_{\rm e }^{[1:K+1]}(\mathbf{r})\big]^{-1} \preccurlyeq	\big[\mathbf{F}_{\rm e }^{[1:K]}(\mathbf{r})\big]^{-1}. 		
\end{equation}
Taking the trace on both sides yields
\begin{equation} \label{EFIM75}
		\Tr\big(\mathbf{F}_{\rm e }^{[1:K+1]}(\mathbf{r})\big)^{-1} \leq \Tr\big(\mathbf{F}_{\rm e }^{[1:K]}(\mathbf{r})\big)^{-1}.
\end{equation}
Recalling that ${\rm CRLB}^{[1:K+1]}(\mathbf{r}) =\Tr\big(\mathbf{F}_{\rm e }^{[1:K+1]}(\mathbf{r})\big)^{-1}$ and $ {\rm CRLB}^{[1:K]}(\mathbf{r})=\Tr\big(\mathbf{F}_{\rm e }^{[1:K]}(\mathbf{r})\big)^{-1}$, the result follows.
\end{proof}

\subsection{Increasing Receivers with Self-Positioning Errors} 

In cellular networks, mobile users can also serve as cooperative receivers.\footnote{In practice, the RIS typically lacks knowledge of the UE's position, and the UE can only engage in cooperative sensing by opportunistically receiving reflected signals from the RIS during movement. To increase the chances of such reflections, the RIS can use random phase configurations or beam-scanning strategies for spatial diversity. Since SHM does not demand real-time sensing, the UE can accumulate these signals over time to achieve cooperative sensing gains.} However, their mobility introduces position uncertainties that must be taken into account. Moreover, even when the receiver is a fixed base station, external environmental disturbances (such as wind or structural vibrations) may cause slight movements or vibrations of the antenna array, which also lead to self-positioning errors. In this subsection, we analyze how such self-positioning errors of the receivers affect the sensing accuracy.

We begin by considering a scenario with a single receiver at a specific time instant, where the receiver's position is subject to errors.
Let $\mathbf{b}_{\rm R}={[x_{\rm R},y_{\rm R},z_{\rm R}]}^{\rm T}$ denote the true position of the cooperative receiver. The measured position, $\hat{\mathbf{b}}_{\rm R}$, is modeled as 
\begin{equation} \label{tr}
	\hat{\mathbf{b}}_{\rm R}=\mathbf{b}_{\rm R}+ \boldsymbol{\epsilon}_{\rm R},
\end{equation}
where $\boldsymbol{\epsilon}_{\rm R}$ represents the position error, which follows a Gaussian distribution with a mean of zero and a diagonal covariance matrix $\mathbf{Q}_{\rm R}$.
Let $\mathbf{J}_{\rm R} = {\partial\boldsymbol{\eta}^{\rm T}}/{\partial \mathbf{b}_{\rm R}}$. Noting the relationship between the derivatives, we have $\mathbf{J}_{\rm R} = -\mathbf{J}$. Consequently, the FIM for the system that accounts for the receiver position error is given by 
\begin{align} \label{FIM9}
	\mathbf{F}(\mathbf{r}, \mathbf{b}_{\rm R})  &=\left[\begin{array}{cc}
		\mathbf{J} \mathbf{F}(\boldsymbol{\eta})  \mathbf{J}^{\rm H}  &  \mathbf{J} \mathbf{F}(\boldsymbol{\eta}) \mathbf{J}_{\rm R}^{\rm H}  \\		
		\mathbf{J}_{\rm R} \mathbf{F}(\boldsymbol{\eta}) \mathbf{J}^{\rm H}              & \mathbf{J}_{\rm R} \mathbf{F}(\boldsymbol{\eta}) \mathbf{J}_{\rm R}^{\rm H}+\mathbf{Q}_{\rm R}^{-1}\\
	\end{array}\right].
\end{align}
Let $\mathbf{U} = [\mathbf{u}_{\tau_{\rm r}}, \mathbf{u}_{\phi_{\rm r}^{\rm az}}, \mathbf{u}_{\phi_{\rm r}^{\rm el}}]$. By performing a unitary transformation on $ \mathbf{Q}_{\rm R}^{-1}$, we can write $\mathbf{Q}_{\rm R}^{-1} = \mathbf{U}\tilde{\mathbf{Q}}_{\rm R}^{-1}\mathbf{U}^{\rm H} $, where 
\begin{align}\label{tildeQ}
	\tilde{\mathbf{Q}}_{\rm R}^{-1} = \left[\begin{array}{ccc}
		\epsilon_{\tau}^{-1}  & 0                                  & 0  \\		
		0                      & \epsilon_{\phi^{\rm az}}^{-1}    & 0  \\
		0                      & 0                                  & \epsilon_{\phi^{\rm el}}^{-1} \\
	\end{array}\right].
\end{align}  
Here, $\epsilon_{\tau}^{-1}$, $\epsilon_{\phi^{\rm az}}^{-1}$, and $\epsilon_{\phi^{\rm el}}^{-1}$ represent the information intensities along the directions of $\mathbf{u}_{\tau_{\rm r}}$, $\mathbf{u}_{\phi_{\rm r}^{\rm az}}$, and $\mathbf{u}_{\phi_{\rm r}^{\rm el}}$, respectively, as shown in Fig. \ref{ellips}(b). 
Based on these developments, we state the following theorem.
\begin{theorem}\label{T1}
The EFIM for 3D RIS localization in the presence of receiver self-positioning errors is given by
\begin{align} \label{EFIM55} 
	\tilde{\mathbf{F}}_{\rm e }(\mathbf{r})& = \frac{\lambda_{\tau_{\rm r}}\epsilon_{\tau}^{-1}}{\lambda_{\tau_{\rm r}}+\epsilon_{\tau}^{-1}}  \mathbf{u}_{\tau_{\rm r}}\mathbf{u}_{\tau_{\rm r}}^{\rm H} \nonumber\\
	& ~~+ \frac{(\lambda_{\phi_{\rm r}^{\rm az}}-\chi_{\phi_{\rm r}^{\rm az}\phi_{\rm r}^{\rm el}})\epsilon_{\phi^{\rm az}}^{-1}}{\lambda_{\phi_{\rm r}^{\rm az}}-\chi_{\phi_{\rm r}^{\rm az}\phi_{\rm r}^{\rm el}}+\epsilon_{\phi^{\rm az}}^{-1}} \mathbf{u}_{\phi_{\rm r}^{\rm az}}\mathbf{u}_{\phi_{\rm r}^{\rm az}}^{\rm H} \nonumber\\
&~~ + \frac{(\lambda_{\phi_{\rm r}^{\rm el}}-\chi_{\phi_{\rm r}^{\rm el}\phi_{\rm r}^{\rm az}})\epsilon_{\phi^{\rm el}}^{-1}}{\lambda_{\phi_{\rm r}^{\rm el}}-\chi_{\phi_{\rm r}^{\rm el}\phi_{\rm r}^{\rm az}}+\epsilon_{\phi^{\rm el}}^{-1}}\mathbf{u}_{\phi_{\rm r}^{\rm el}}\mathbf{u}_{\phi_{\rm r}^{\rm el}}^{\rm H}.
\end{align}
\end{theorem}
The proof is provided in Appendix \ref{A2}.

Theorem \ref{T1} reveals that when a cooperative receiver exhibits large position errors, i.e.,  $\epsilon_{\tau}^{-1}$, $\epsilon_{\phi^{\rm az}}^{-1}$, $\epsilon_{\phi^{\rm el}}^{-1} \rightarrow 0$, the EFIM reduces to $\tilde{\mathbf{F}}_{\rm e }(\mathbf{r})=\mathbf{0}$. This implies that the cooperative receiver provides no useful information.
Building upon the results from Section~\ref{mr}, we state the following corollary.
\begin{corollary}\label{CO2}	
	Introducing an additional receiver with limited position uncertainty can still reduce the CRLB. Furthermore, the CRLB satisfies the following inequality:
	\begin{multline} \label{eq:th4} 		
		{\rm CRLB}^{[1:K+1]}(\mathbf{r})  \\
		\overset{(a)}{\leq}  {\rm CRLB}^{[1:K+1]}(\mathbf{r},\mathbf{b}_{\rm R}) \\
		 \overset{(b)}{\leq} {\rm CRLB}^{[1:K]}(\mathbf{r}),
	\end{multline}
	where
	\begin{equation}
	{\rm CRLB}^{[1:K+1]}(\mathbf{r},\mathbf{b}_{\rm R})= \Tr\Big(\big(\mathbf{F}_{\rm e }^{[1:K]}(\mathbf{r})+\tilde{\mathbf{F}}_{\rm e }^{[K+1]}(\mathbf{r})\big)^{-1}\Big),
    \end{equation}
    and $\tilde{\mathbf{F}}_{\rm e }^{[K+1]}(\mathbf{r})$ is obtained from \eqref{EFIM55} by replacing the parameters $ \lambda_{\tau_{\rm r}} $, $\lambda_{\phi_{\rm r}^{\rm az}}$, $\chi_{\phi_{\rm r}^{\rm az}\phi_{\rm r}^{\rm el}}$, $\lambda_{\phi_{\rm r}^{\rm el}}$, and $\chi_{\phi_{\rm r}^{\rm el}\phi_{\rm r}^{\rm az}}$ with their respective values for the ${(K+1)}$-th receiver.  ${\rm CRLB}^{[1:K+1]}(\mathbf{r})$ and ${\rm CRLB}^{[1:K]}(\mathbf{r})$ represent the CRLB obtained from the cooperation of ${K+1}$ and $K$ receivers with perfectly known positions, respectively. 
\end{corollary}
\begin{proof}
	According to \eqref{EFIM55}, we can infer that $\mathbf{0}  \preccurlyeq \tilde{\mathbf{F}}_{\rm e }^{[K+1]}(\mathbf{r}) \preccurlyeq {\mathbf{F}}_{\rm e }^{[K+1]}(\mathbf{r})$. Thus, inequalities (a) and (b) follow directly.
\end{proof}

Corollary~\ref{CO2} indicates that a cooperative receiver with large positioning errors does not degrade the overall Fisher information of the network.

\begin{figure}
	\centering
	\includegraphics[width=0.42\textwidth]{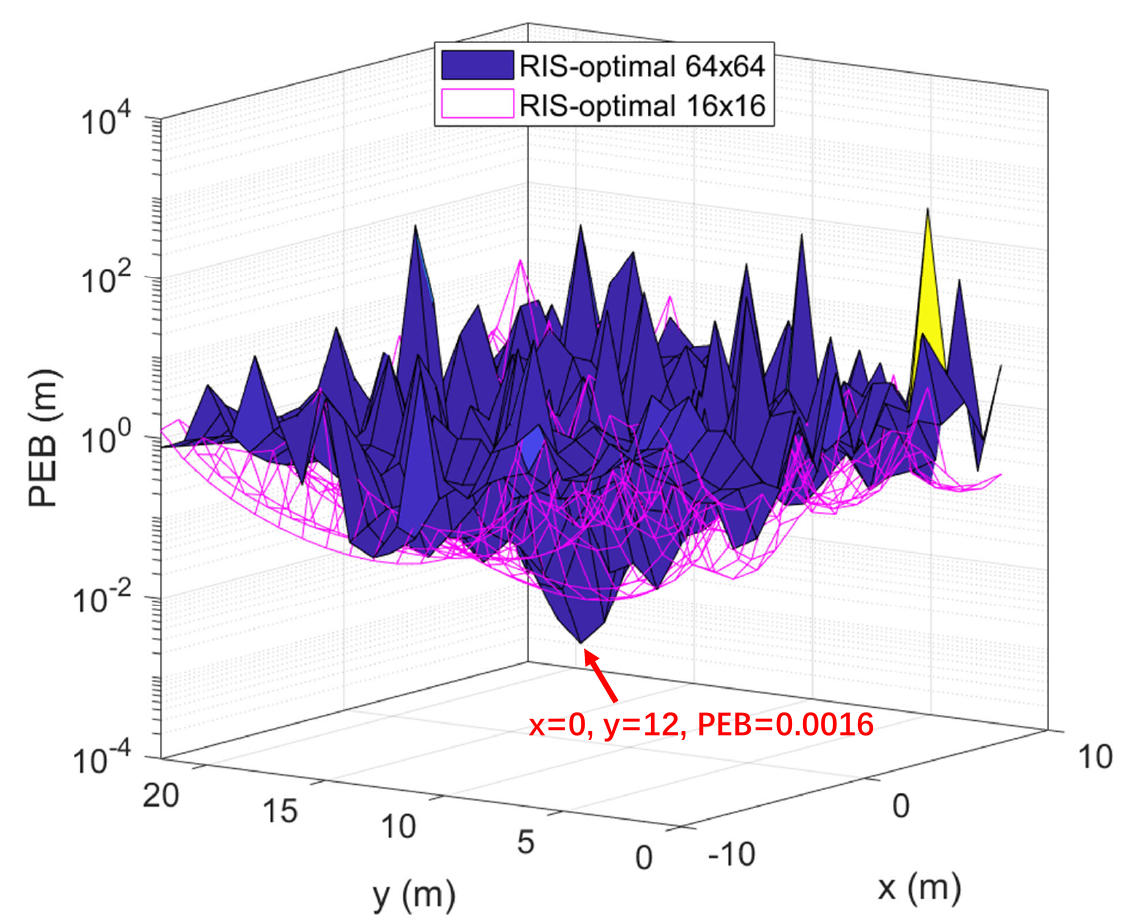}
	\caption{PEB variations with respect to receiver position changes.}
	\label{fig3}
\end{figure}
 
\begin{figure}
	\centering
	\includegraphics[width=0.495\textwidth]{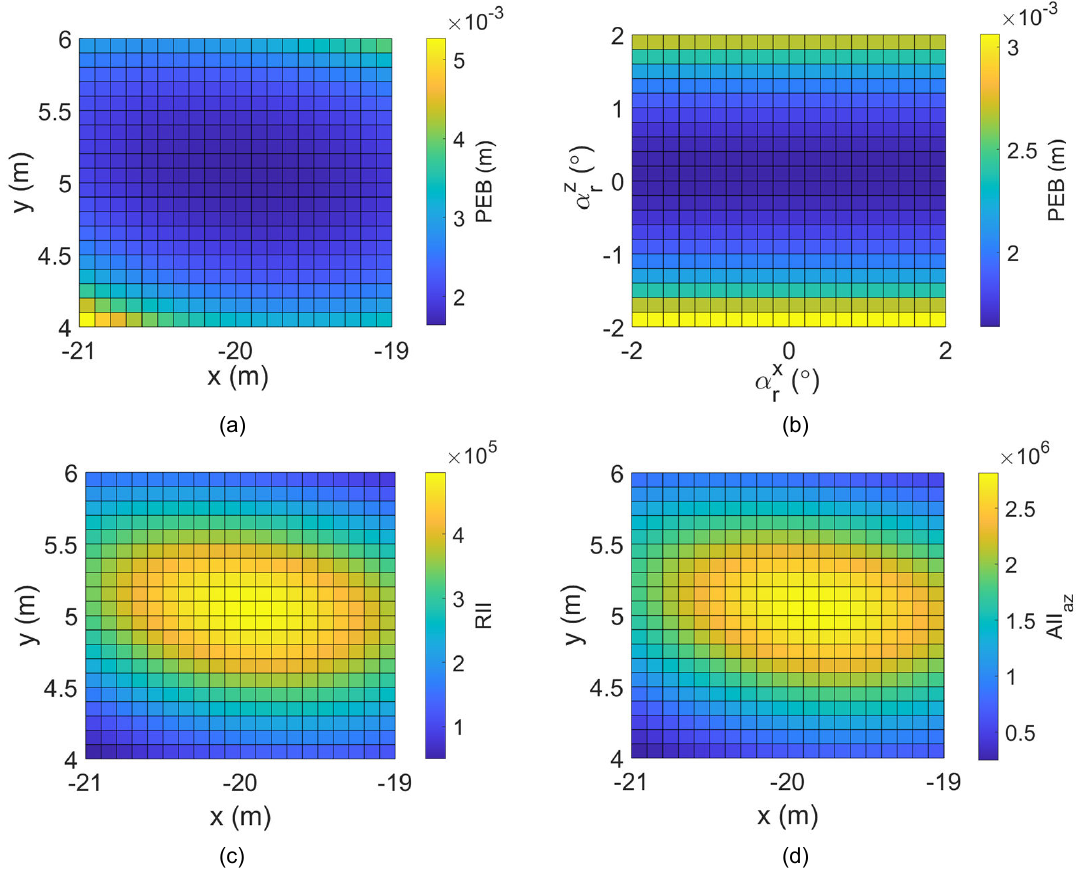}
	\caption{PEB and information intensity variations with respect to changes in RIS positions. AII$_{\rm el}$ shares the same trend as AII$_{\rm az}$ and is thus omitted.}
	\label{fig4}
\end{figure}
 
\begin{figure*}
	\hspace{0.95cm}
	\includegraphics[width=0.92\textwidth]{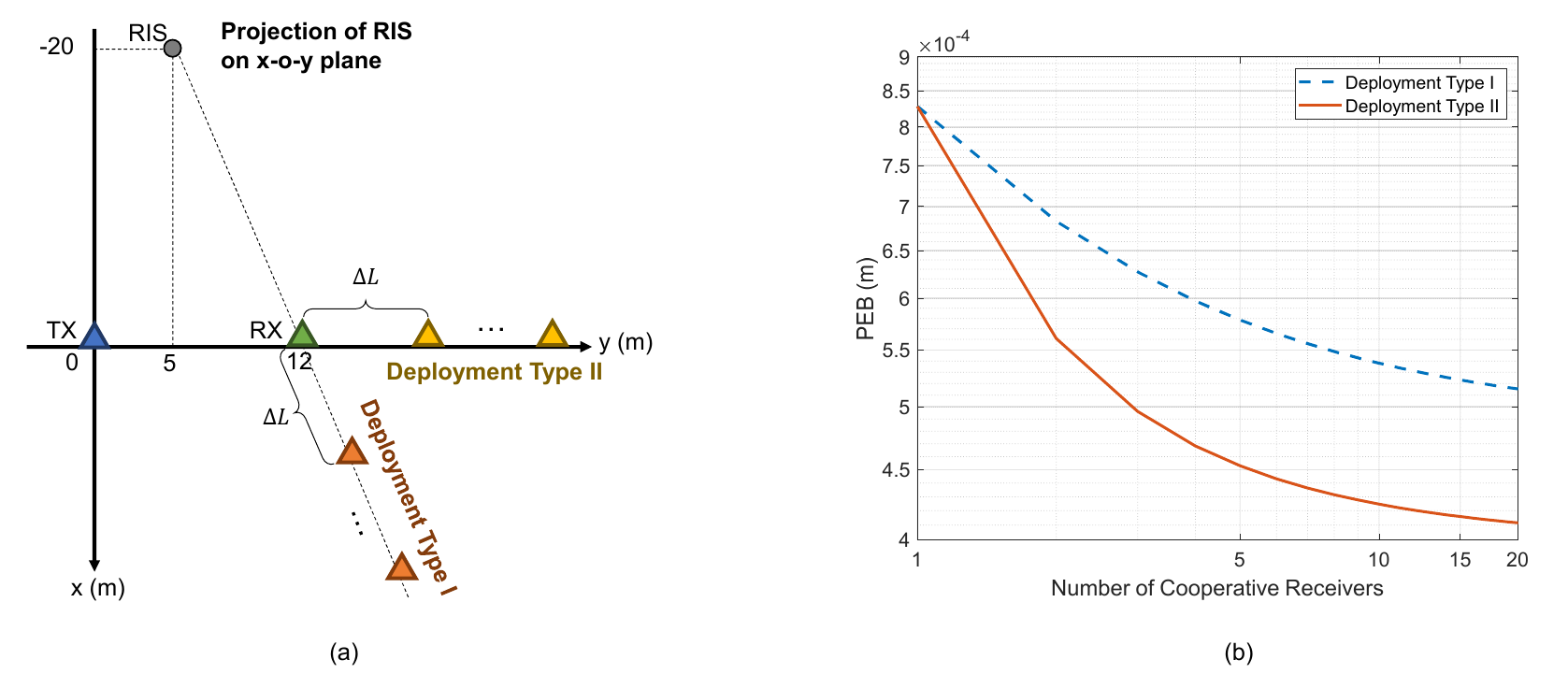}
	\caption{PEB performance in spatial cooperation. (a) Illustration of Deployment Types I and II. (b)Relationship between PEB and the number of cooperative receivers under different deployment strategies.  }
	\label{fig5}
\end{figure*}

\section{Numerical Results}\label{sec-SIM}

\begin{figure}
	\centering
	\includegraphics[width=0.42\textwidth]{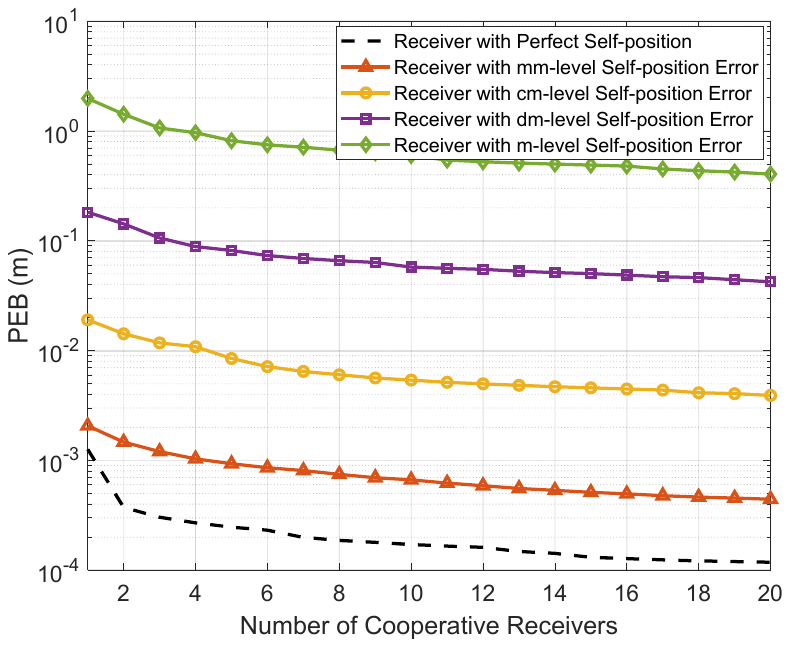}
	\caption{The impact of cooperative receiver self-positioning error on PEB performance, where m, dm, cm, and mm are abbreviations for meter, decimeter, centimeter, and millimeter, respectively.}
	\label{fig6}
\end{figure}
\begin{figure*}
	\center
	\includegraphics[width=1\textwidth]{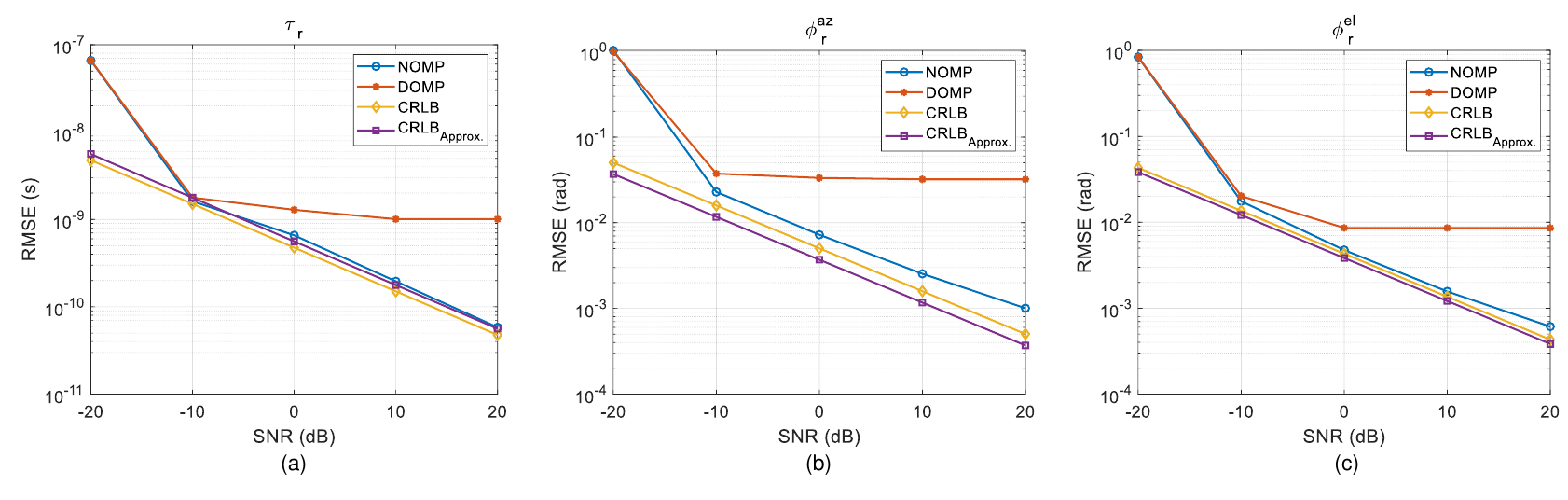} 
	\caption{RMSE of the delay, azimuth AOA, and elevation AOA estimation as a function of SNR. The estimators of DOMP and NOMP, and CRLB are considered.}
	\label{NOMP}
\end{figure*}

\subsection{Comparison of RII and AII Characteristics}
In this subsection, we analyze the similarities and differences between RII and AII. The azimuth and elevation elements of both the transmitter and receiver, $N_{\rm T}^{\rm az}$, $N_{\rm T}^{\rm el}$, $N_{\rm R}^{\rm az}$, and $N_{\rm R}^{\rm el}$, are set to $4$. The RIS size $M$ is set to $64 \times 64$, the subcarrier spacing is 120 kHz, and the SNR is 20 dB.
The coordinates are set as follows: $\mathbf{b}_{\rm T}=[0,0,0]^{\rm T}$ meters (m), $\mathbf{b}_{\rm R}={\rm scale}\times[0,12,0]^{\rm T}$ m, $\mathbf{r}={\rm scale}\times[-20,5,10]^{\rm T}$ m, and $\boldsymbol{\alpha}_{\rm r}=[0^{\circ},0^{\circ}]$. Unless otherwise stated, the observation time instants $T$  is set to 2. The value of ${\rm scale}$ ranges from 10 to 100, and the number of subcarriers is set to 1, 10, 100, and 1000.
As shown in Fig. \ref{fig1}(a), RII remains constant with respect to scene size and scales quadratically with bandwidth, whereas AII decreases as the scene size increases. Fig. \ref{fig1}(b) illustrates that the PEB, which is derived from the fusion of RII and AII, is range-limited when the number of subcarriers is small ($N_{\rm c}=1$) and converges to being angle-limited as $N_{\rm c}$ increases.

\subsection{Tolerance to Deployment Errors}
First, we analyze how PEB fluctuates with changes in receiver positions after deploying the RIS with phases optimized based on the initial coordinates: $\mathbf{b}_{\rm T}=[0,0,0]^{\rm T}$ m, $\mathbf{b}_{\rm R}=[0,12,0]^{\rm T}$ m, $\mathbf{r}=[-20,5,10]^{\rm T}$ m, and $\boldsymbol{\alpha}_{\rm r}=[0^{\circ},0^{\circ}]$. 
The number of subcarriers is set to $N_{\rm c}=100$, while other parameters remain unchanged. Fig. \ref{fig3} illustrates the variations in PEB with respect to receiver positions for RIS sizes of $M = 64 \times 64$ and $M = 16 \times 16$, respectively. The receiver's movement area is defined by $x_{\mathbf{b}_{\rm r}} \in [-10,10]$ m and $y_{\mathbf{b}_{\rm r}} \in [0,25]$ m. 
The results indicate that with an optimally configured RIS phase for specific receiver positions, millimeter-level PEB can be achieved, enabling precise monitoring of structural deformations. However, maintaining the same RIS phase configuration while altering receiver positions leads to significant degradation in both information intensity and PEB performance, particularly with larger RIS element arrays. This highlights the necessity of reconfiguring the RIS phase according to new receiver positions. Additionally, when multiple receivers cooperate, optimal RIS phase configuration through time-division multiplexing is essential.

Second, we evaluate the impact of slight variations in position and orientation of the RIS on information intensity and PEB, assuming a fixed RIS phase configuration (Fig. \ref{fig4}). The RIS size $M$ is $64 \times 64$, with all other parameters unchanged. The RIS position variations are set within the range $x_{\rm r} \in [-21,-19]$ m and $y_{\rm r} \in [4,6]$ m. 
Simulation results reveal that slight RIS position deviations (within $\pm 1$ m of the true position) do not significantly impact the PEB [Fig. \ref{fig4}(a)] and information intensity [Fig. \ref{fig4}(c) and (d)]. 
Keeping RIS position $\mathbf{r}=[-20,5,10]^{\rm T}$ m unchanged, we vary the orientation within the range $\alpha^x_{\rm r}, \alpha^z_{\rm r} \in [-2^{\circ},2^{\circ}]$. 
$\alpha^x_{\rm r}$ and $\alpha^z_{\rm r}$ represent the rotation angles around the $x$-axis and $z$-axis, respectively. It can be observed in  Fig. \ref{fig4}(b) that rotation around the $x$-axis has little effect on PEB performance, while rotation around the $z$-axis causes a slight degradation in PEB.
This suggests that the RIS phase configuration can be optimized at the time of deployment and maintained thereafter, enabling robust performance in detecting small structural deformations even when slight self-positional changes occur due to structural variations or environmental disturbances such as wind.

\subsection{PEB Performance Analysis for Cooperation}
First, we investigate the impact of different cooperative receiver deployments on PEB performance. Two deployment types, Type I and Type II, are illustrated in Fig. \ref{fig5}(a).
\begin{itemize}  
    \item Deployment Type I: The cooperative receivers are placed along the projection of the RIS onto the x-o-y plane and the line connecting  the RX.    
    \item Deployment Type II: Cooperative receivers are arranged along the positive y-axis.
\end{itemize}
In both deployment types, the receivers are spaced at a uniform distance of  $\Delta L =20$ m.

The RIS size is fixed at $M=64 \times 64$, the subcarrier number is set to $N_{\rm c}=1000$, and other parameters remain unchanged. The results indicate that for the same number of cooperative receivers, Deployment Type II outperforms Deployment Type I in terms of PEB performance [Fig. \ref{fig5}(b)]. This improvement is attributed to the enhanced spatial diversity provided by Deployment Type II. Thus, the placement of cooperative receivers is crucial for improving sensing performance, and the proposed Corollary \ref{C3} provides guidelines for optimizing their deployment.

Next, keeping the simulation parameters unchanged, we evaluate PEB performance under different levels of receiver self-positioning errors, ranging from millimeter-level to m-level (Fig. \ref{fig6}). In this simulation, the positions of cooperative receivers are randomly distributed within a circle of radius 20 m centered at $[0,0]$ on the x-o-y plane.
The black dashed line in Fig. \ref{fig6} represents the PEB performance when cooperative receiver positions are perfectly known. Under the current simulation conditions, cooperation among three receivers with millimeter-level self-positioning errors can achieve the same PEB performance as a perfectly self-positioned receiver. Furthermore, cooperation among receivers with decimeter-, centimeter-, and millimeter-level self-positioning errors can achieve sub-decimeter, sub-cm, and sub-m-level PEB performance, respectively, catering to varying deformation detection accuracy requirements. Increasing the observation time instants, the number of antenna elements, and bandwidth can further enhance PEB performance.

To validate the accuracy of the EFIM derived in the proposed theorem, we apply the NOMP algorithm \cite{NOMP} to extract delay and angle parameters from the received signal given in \eqref{y-s}, and compute the root mean square error (RMSE) of the estimates.
The corresponding CRLBs are calculated based on the approximate analytical expressions in \eqref{fim-e2}, \eqref{fim-az}, \eqref{fim-el}, and \eqref{fim-elaz}.
In the simulations, the number of receiver antennas $N_{\rm R}^{\rm az}$ and $N_{\rm R}^{\rm el}$ is set to 8.
The RIS size $M$ is set to $64 \times 64$, and the bandwidth is 30.72 MHz.
The coordinates are set as follows: $\mathbf{b}{\rm R} = [0,0,0]^{\rm T}$ m, $\mathbf{b}_{\rm T} = [20,0,0]^{\rm T}$ m, $\mathbf{r} = [15,12,12]^{\rm T}$ m, and $\boldsymbol{\alpha}_{\rm r} = [0^{\circ}, 0^{\circ}]$.
The simulation results are illustrated in Fig. \ref{NOMP}, where ``CRLB'' denotes the numerical solution and ``CRLB$_{\rm Approx.}$'' represents the derived closed-form expression. It can be observed that the delay and angle estimates obtained via the NOMP algorithm closely approach the CRLB$_{\rm Approx.}$, and the numerical and closed-form CRLBs are also in close agreement.
The consistency between the RMSEs and CRLBs confirms the correctness of the EFIM derivation.

\subsection{Detection Probability Performance for SHM}

\begin{figure}
	\centering
	\includegraphics[width=0.42\textwidth]{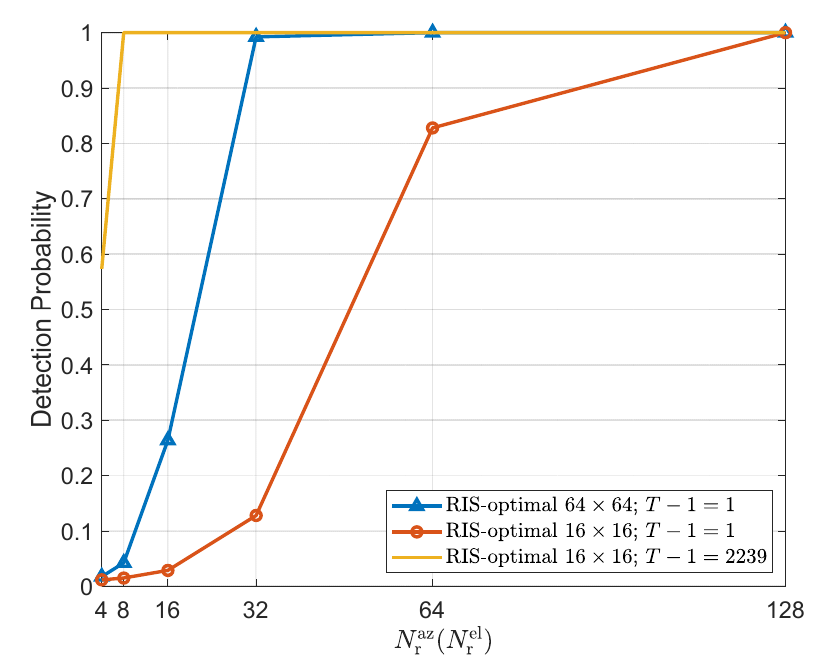} 
	\caption{Detection probability performance for millimeter-level deformations under various parameter configurations.}
	\label{fig7}
\end{figure}

In this subsection, we employ the standard constant false alarm rate method to evaluate the detection probability performance for millimeter-level deformations under various parameter configurations, such as the number of receiving antennas, RIS elements, and observation time instants. This analysis highlights the capability of the cellular system in detecting slow and subtle changes in structures.
In this expanded simulation scenario, the receiver position is set to $\mathbf{b}_{\rm R} = [0,120,0]^{\rm T}$ m and the RIS position to $\mathbf{r} = [-200,50,100]^{\rm T}$ m. Here, we do not consider cooperation among multiple receivers. The number of subcarriers $N_{\rm c}$ is set to 1000, and the SNR is maintained at 20 dB.
As illustrated in Fig. \ref{fig7}, achieving a 100\% detection probability with a 
$16 \times 16$ RIS requires increasing the number of receiving antennas to 
$128\times128$. 
In accordance with the 5G NR standard, a 20 ms radio frame provides 2240 OFDM symbols available for measurement.
Given the assumption that no significant structural changes occur within this 20 ms interval, the simulation demonstrates that when $T=2240$, a 100\% detection probability can be achieved using a $16 \times 16$ RIS with $8 \times 8$ receiving antennas.
This result suggests that achieving high detection probability for millimeter-level deformations with a limited number of receiving antennas can be accomplished by either increasing the number of RIS elements or extending the observation time.

In addition, we evaluate the detection probability performance for various levels of building deformations under different parameter configurations (varying SNR and number of RIS elements) to showcase the capability of the cellular system in identifying structural states. As depicted in Fig. \ref{fig55}, a 100\% detection probability is achieved with a $4 \times 4$ RIS at an SNR of $-20$ dB for significant building deformations (meter level). However, to attain a high detection probability (96\%) for subtle deformations (millimeter level), an increased number of RIS elements ($64 \times 64$) and a higher received SNR (20 dB) are required.

\begin{figure*}
	\centering
	\includegraphics[width=1\textwidth]{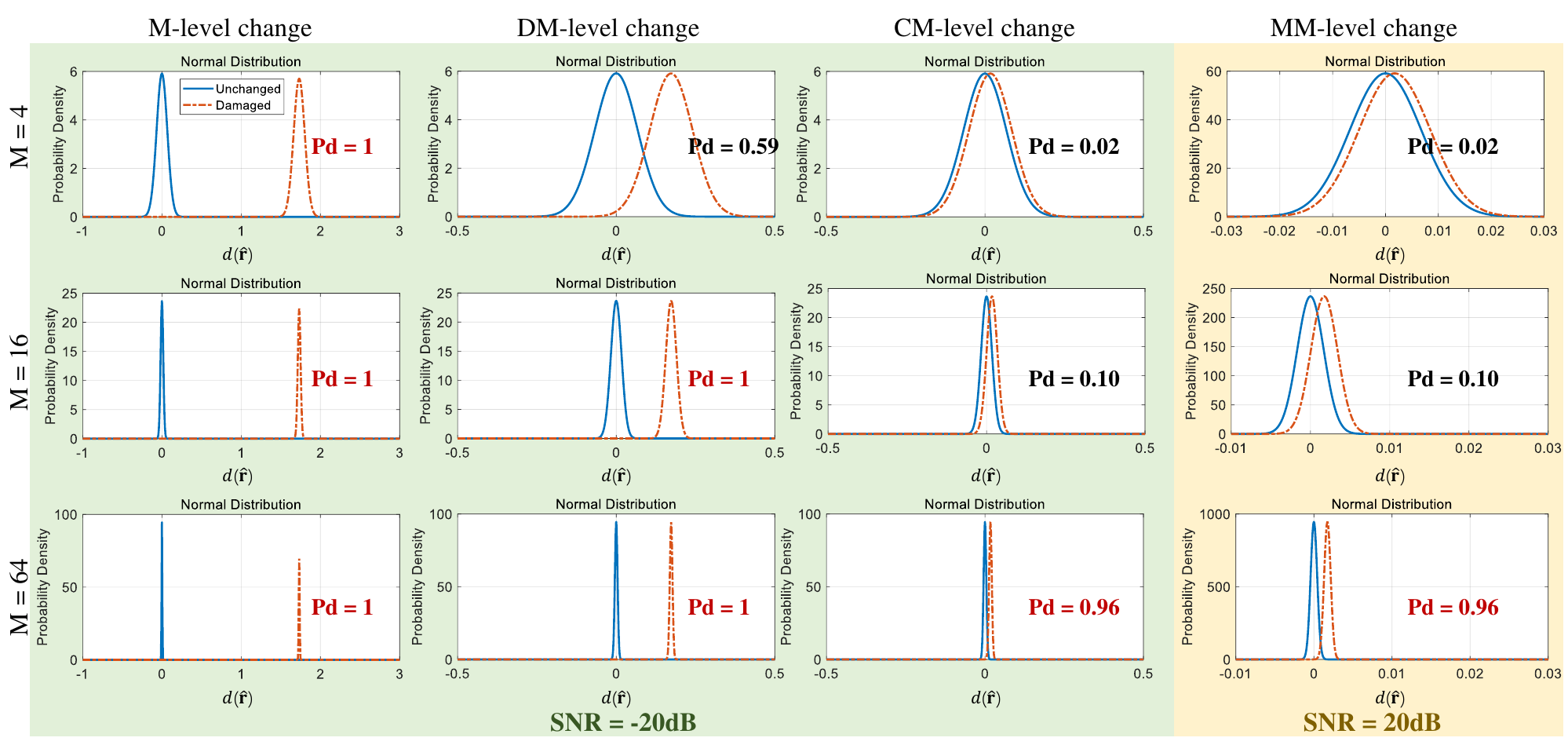}\caption{Detection probability performance across various building deformation levels under different parameter configurations, including varying SNR and number of RIS elements.}
	\label{fig55}
\end{figure*}

\section{Conclusions}\label{sec-conclusion}
This study proposed a framework to investigate the potential of cellular systems in SHM applications.
Specifically, we introduced RISs as reference points for monitoring building structures to suppress clusters and interference. Then, based on the proposed theory of the information ellipsoid, we analyzed spatiotemporal cooperation in 3D cellular networks to meet the stringent sensing accuracy requirements in SHM.
Through theoretical and numerical analyses, we demonstrated that the optimal RIS phase can be configured based on the initial installation, and slight RIS position and orientation variations will not significantly degrade the PEB. Additionally, cellular systems can detect millimeter-level deformations, even with self-positioning errors in cooperative nodes, by increasing observation time, the number of nodes, optimizing their deployment and RIS phases, and enhancing system bandwidth or antennas.
Future research should focus on building rotation estimation, classification of various deformation types, multi-RIS cooperation, multi-modal fusion, large-aperture RIS and near-field effects, as well as experimental validation.



\appendix
	\subsection{Geometric Relationship}\label{A0}
	According to the system model, the RIS position is given by $\mathbf{r} = [x_{\rm r},y_{\rm r},z_{\rm r}]^{\rm T}$, and the receiver position is $\mathbf{b}_{\rm R}=[x_{\rm R},x_{\rm R},x_{\rm R}]^{\rm T}$. Their geometric relationship is as follows: 
	\begin{subequations}
		\begin{align}
		x_{\rm r} - x_{\rm R} &= c(\tau_{\rm r}-\tau_{\rm tr})\cos \phi_{\rm r}^{\rm el} \cos\phi_{\rm r}^{\rm az}, \label{GR0}\\ 
		y_{\rm r} - y_{\rm R} &= c(\tau_{\rm r}-\tau_{\rm tr})\cos \phi_{\rm r}^{\rm el} \sin\phi_{\rm r}^{\rm az},\label{GR1}\\
		z_{\rm r} - z_{\rm R} &= c(\tau_{\rm r}-\tau_{\rm tr})\sin \phi_{\rm r}^{\rm el}\label{GR2}.
		\end{align}
	\end{subequations}
	Given the RIS orientation $\boldsymbol{\alpha}_{\rm r}=[ \alpha_{\rm r}^{\rm x},\alpha_{\rm r}^{\rm z}]^{\rm T}$, we define a local coordinate system where the RIS serves as the origin and the RIS array plane aligns with the $y'$-$z'$ plane. In this coordinate system, the receiver position is represented as $\mathbf{b}'_{\rm R}=[x'_{\rm R},y'_{\rm R},z'_{\rm R}]^{\rm T}$, which satisfies
	\begin{equation}\label{GR3}
    \mathbf{b}'_{\rm R}= -\mathbf{R}^{-1}( \alpha_{\rm r}^{\rm x},\alpha_{\rm r}^{\rm z})(\mathbf{r}-\mathbf{b}_{\rm R}),
	\end{equation}
	where the rotation matrix is given by
	\begin{equation}\label{GR4}
    \mathbf{R}( \alpha_{\rm r}^{\rm x},\alpha_{\rm r}^{\rm z})\!\!=\!\!\left[\!\!\begin{array}{ccc}
    \cos\alpha_{\rm r}^{\rm z} & -\sin\alpha_{\rm r}^{\rm z}\cos \alpha_{\rm r}^{\rm x} & -\sin\alpha_{\rm r}^{\rm z} \sin\alpha_{\rm r}^{\rm x} \\
    \sin\alpha_{\rm r}^{\rm z} & \cos\alpha_{\rm r}^{\rm z}\cos\alpha_{\rm r}^{\rm x} & \cos\alpha_{\rm r}^{\rm z}\sin\alpha_{\rm r}^{\rm x}\\
    0 & -\sin\alpha_{\rm r}^{\rm x} & \cos\alpha_{\rm r}^{\rm x}
    \end{array}\!\!\right]\!\!,
	\end{equation}
	with $\mathbf{R}^{-1}( \alpha_{\rm r}^{\rm x},\alpha_{\rm r}^{\rm z}) = \mathbf{R}^{\rm T}( \alpha_{\rm r}^{\rm x},\alpha_{\rm r}^{\rm z}) $.
	Furthermore, we have
	\begin{subequations}\label{GR5}
		\begin{align}
		x'_{\rm R} &= c(\tau_{\rm r}-\tau_{\rm tr})\cos \tilde{\phi}_{\rm r}^{\rm el} \cos\tilde{\phi}_{\rm r}^{\rm az}, \\ 
		y'_{\rm R} &= c(\tau_{\rm r}-\tau_{\rm tr})\cos \tilde{\phi}_{\rm r}^{\rm el} \sin\tilde{\phi}_{\rm r}^{\rm az},\\
		z'_{\rm R} &= c(\tau_{\rm r}-\tau_{\rm tr})\sin \tilde{\phi}_{\rm r}^{\rm el}.
		\end{align}
	\end{subequations}	
    By combining \eqref{GR3} and \eqref{GR5}, we establish the relationship between the AOD measurement at the RIS, the RIS state, and the receiver position. A similar geometric relationship holds between the RIS and the transmitter.

	\subsection{Information Ellipsoid}\label{B}
	Take the partial derivative with respect to $x_{\rm r}$, $y_{\rm r}$, $z_{\rm r}$ on both sides of \eqref{GR0}, \eqref{GR1}, \eqref{GR2}, respectively, we have
	\begin{subequations}\label{pd}
		\begin{align}
		\!	\frac{\partial \tau_{\rm r} }{\partial \mathbf{r}} \!&\!=\! \frac{1}{c} \left[ \cos\phi_{\rm r}^{\rm az}\cos \phi_{\rm r}^{\rm el}, \sin\phi_{\rm r}^{\rm az} \cos \phi_{\rm r}^{\rm el},  \sin \phi_{\rm r}^{\rm el} \right]^{\rm T},\\
		\!	\frac{\partial \phi_{\rm r}^{\rm az} }{\partial \mathbf{r}} \!&\!=\! \frac{1}{c(\tau_{\rm r}-\tau_{\rm tr})\cos \phi_{\rm r}^{\rm el}}  \left[-\sin \phi_{\rm r}^{\rm az},\cos\phi_{\rm r}^{\rm az},0 \right]^{\rm T}, \label{paz}\\
		\!	\frac{\partial \phi_{\rm r}^{\rm el} }{\partial \mathbf{r}} \!&\!=\! \frac{1}{c(\tau_{\rm r}\!-\!\tau_{\rm tr})}\!\! \left[ - \cos\!\phi_{\rm r}^{\rm az}\!\sin\! \phi_{\rm r}^{\rm el}, -\sin\!\phi_{\rm r}^{\rm az}\!\sin\! \phi_{\rm r}^{\rm el},\cos\! \phi_{\rm r}^{\rm el}\right]\!^{\rm T}\label{pel}\!.
		\end{align}
	\end{subequations}	    
	Therefore, we have 
	\begin{equation}
		\mathbf{J}_{\rm A}=\Big[\frac{\partial \tau_{\rm r} }{\partial \mathbf{r}} ,\frac{\partial \phi_{\rm r}^{\rm az} }{\partial \mathbf{r}}, \frac{\partial \phi_{\rm r}^{\rm el} }{\partial \mathbf{r}}\Big]. 
	\end{equation}

Defining $\mathbf{u}_{\tau_{\rm r}}$, $\mathbf{u}_{\phi_{\rm r}^{\rm az}}$, and $\mathbf{u}_{\phi_{\rm r}^{\rm el}}$ in accordance with \eqref{unitv}, 
then, we have the unitary matrix $\mathbf{U} = [\mathbf{u}_{\tau_{\rm r}}, \mathbf{u}_{\phi_{\rm r}^{\rm az}}, \mathbf{u}_{\phi_{\rm r}^{\rm el}}]$.
Based on \eqref{D1FIMA} and \eqref{D1FIM4}, we derive 
\begin{align}\label{expandFIM}
\mathbf{F}(\mathbf{r}) &\!=\!\mathbf{U}\!\! \left[\!\begin{array}{ccc}
	\!\!	\frac{F_{ \tau_{\rm r} \tau_{\rm r}}}{c^2} & \!\!\!\!0 & \!\!\!\!0 \\
		0 & \!\!\!\! \frac{F_{  \phi_{\rm r}^{\rm az} \phi_{\rm r}^{\rm az}} }{c^2(\tau_{\rm r}-\tau_{\rm tr})^2\cos^2 \phi_{\rm r}^{\rm el}}  & \!\!\!\!\frac{F_{  \phi_{\rm r}^{\rm el} \phi_{\rm r}^{\rm az}}  }{c^2(\tau_{\rm r}-\tau_{\rm tr})^2\cos\phi_{\rm r}^{\rm el}} \\
		0 &\!\!\!\!\frac{F_{  \phi_{\rm r}^{\rm az} \phi_{\rm r}^{\rm el}}  }{c^2(\tau_{\rm r}-\tau_{\rm tr})^2\cos\phi_{\rm r}^{\rm el}}&\!\!\!\!\frac{F_{  \phi_{\rm r}^{\rm el} \phi_{\rm r}^{\rm el}} }{c^2(\tau_{\rm r}-\tau_{\rm tr})^2}
	\end{array}\right] \!\! \mathbf{U}^{\rm H}.
\end{align}
From the EFIM analysis \cite{EFIM}, \cite{ELIP}, we obtain \eqref{expandEFIM}. 
\begin{figure*}
	\vspace{-0.2cm}
	\begin{align}\label{expandEFIM}
		\mathbf{F}_{\rm e}(\mathbf{r}) &=\mathbf{U}\left[\begin{array}{ccc}
			\!\!\frac{F_{ \tau_{\rm r} \tau_{\rm r}}}{c^2} & 0 & 0 \\
			0 & \!\! \frac{F_{  \phi_{\rm r}^{\rm az} \phi_{\rm r}^{\rm az}} }{c^2(\tau_{\rm r}-\tau_{\rm tr})^2\cos^2 \phi_{\rm r}^{\rm el}} - \frac{F_{  \phi_{\rm r}^{\rm el} \phi_{\rm r}^{\rm az}}F_{  \phi_{\rm r}^{\rm az} \phi_{\rm r}^{\rm el}} }{c^2(\tau_{\rm r}-\tau_{\rm tr})^2\cos^2 \phi_{\rm r}^{\rm el}F_{  \phi_{\rm r}^{\rm el} \phi_{\rm r}^{\rm el}}} & 0 \\
			0 & 0 & \!\! \frac{F_{  \phi_{\rm r}^{\rm el} \phi_{\rm r}^{\rm el}} }{c^2(\tau_{\rm r}-\tau_{\rm tr})^2} - \frac{F_{  \phi_{\rm r}^{\rm az} \phi_{\rm r}^{\rm el}} F_{  \phi_{\rm r}^{\rm el} \phi_{\rm r}^{\rm az}}}{c^2(\tau_{\rm r}-\tau_{\rm tr})^2F_{  \phi_{\rm r}^{\rm az} \phi_{\rm r}^{\rm az}}}
		\end{array}\!\!\right] \mathbf{U}^{\rm H}.
	\end{align}
	\vspace{-0.2cm}
\end{figure*}
According to the definition of the EFIM \cite{EFIM}, \cite{ELIP}, 
we have
\begin{equation}
	\tr(\mathbf{F}^{-1}(\mathbf{r})) = \tr(\mathbf{F}^{-1}_{\rm e}(\mathbf{r})).
\end{equation}
In the following, we demonstrate that the expressions given in \eqref{expandFIM} and \eqref{expandEFIM} satisfy the above definition.	
Let 
\begin{align}\label{A1}
	a &= \frac{F_{ \tau_{\rm r} \tau_{\rm r}}}{c^2} \\
	b &= \frac{F_{  \phi_{\rm r}^{\rm az} \phi_{\rm r}^{\rm az}} }{c^2(\tau_{\rm r}-\tau_{\rm tr})^2\cos^2 \phi_{\rm r}^{\rm el}}\\ 
	c &= \frac{F_{  \phi_{\rm r}^{\rm el} \phi_{\rm r}^{\rm el}} }{c^2(\tau_{\rm r}-\tau_{\rm tr})^2}\\
	d &= \frac{F_{  \phi_{\rm r}^{\rm el} \phi_{\rm r}^{\rm az}}  }{c^2(\tau_{\rm r}-\tau_{\rm tr})^2\cos\phi_{\rm r}^{\rm el}}. 
\end{align}
Then, we define 
\begin{align}\label{A2}
	\mathbf{A} &=\left[\!\begin{array}{ccc}
		a & 0 & 0 \\
		0 & b & d \\
		0 & d & c
	\end{array}\right], 
\end{align}	
and
\begin{align}\label{B1}
	\mathbf{B} &= \left[\begin{array}{ccc}
		a & 0 & 0 \\
		0 & b-\dfrac{d^2}{c} & 0 \\
		0 & 0 &  c-\dfrac{d^2}{b}
	\end{array}\!\!\right].
\end{align}	
The inverse can be easily derived 
\begin{align}\label{invA}
	\mathbf{A}^{-1} &=\left[\!\begin{array}{ccc}
		a^{-1} & 0 & 0 \\
		0 & \dfrac{c}{bc-d^2} & \dfrac{-d}{bc-d^2} \\
		0 & \dfrac{-d}{bc-d^2} & \dfrac{b}{bc-d^2}
	\end{array}\right], 
\end{align}	
and
\begin{align}\label{invB}
	\mathbf{B}^{-1} &= \left[\begin{array}{ccc}
		a^{-1} & 0 & 0 \\
		0 & \dfrac{c}{bc-d^2} & 0 \\
		0 & 0 &  \dfrac{b}{bc-d^2}
	\end{array}\!\!\right].
\end{align}				
Utilizing the unitary property of the matrix, i.e., $\mathbf{U}^{\rm H} \mathbf{U} = \mathbf{I}$, we have
\begin{equation}
	\tr(\mathbf{F}^{-1}(\mathbf{r})) = \tr(\mathbf{U}\mathbf{A}^{-1}\mathbf{U}^{\rm H})=\tr(\mathbf{A}^{-1}),
\end{equation}
and
\begin{equation}
	\tr(\mathbf{F}^{-1}_{\rm e}(\mathbf{r})) = \tr(\mathbf{U}\mathbf{B}^{-1}\mathbf{U}^{\rm H})=\tr(\mathbf{B}^{-1}).
\end{equation}
Acoording to \eqref{invA} and \eqref{invB}, we have $\tr(\mathbf{A}^{-1})=\tr(\mathbf{B}^{-1})$, thus, the expression of $\mathbf{F}_{\rm e}(\mathbf{r})$ provided in \eqref{expandEFIM} holds true.

Defining
\begin{subequations}\label{define}
	\begin{align}
		\lambda_{\tau_{\rm r}} &= \frac{F_{\tau_{\rm r}\tau_{\rm r}}}{c^2},  \\
		\lambda_{\phi_{\rm r}^{\rm az}}  &= \frac{F_{\phi_{\rm r}^{\rm az}\phi_{\rm r}^{\rm az}}}{c^2(\tau_{\rm r}-\tau_{\rm tr})^2\cos^2\phi_{\rm r}^{\rm el}}, \\
		\lambda_{\phi_{\rm r}^{\rm el}} &= \frac{F_{\phi_{\rm r}^{\rm el}\phi_{\rm r}^{\rm el}}}{c^2(\tau_{\rm r}-\tau_{\rm tr})^2}, \\
		\chi_{\phi_{\rm r}^{\rm az}\phi_{\rm r}^{\rm el}} &= \frac{F_{  \phi_{\rm r}^{\rm el} \phi_{\rm r}^{\rm az}}F_{  \phi_{\rm r}^{\rm az} \phi_{\rm r}^{\rm el}} }{c^2(\tau_{\rm r}-\tau_{\rm tr})^2\cos^2 \phi_{\rm r}^{\rm el}F_{  \phi_{\rm r}^{\rm el} \phi_{\rm r}^{\rm el}}}, \\
		\chi_{\phi_{\rm r}^{\rm el}\phi_{\rm r}^{\rm az}} &=\frac{F_{  \phi_{\rm r}^{\rm az} \phi_{\rm r}^{\rm el}} F_{  \phi_{\rm r}^{\rm el} \phi_{\rm r}^{\rm az}}}{c^2(\tau_{\rm r}-\tau_{\rm tr})^2F_{  \phi_{\rm r}^{\rm az} \phi_{\rm r}^{\rm az}}},
	\end{align}
\end{subequations}
we have
\begin{align}
	\mathbf{F}_{\rm e}(\mathbf{r})& = \lambda_{\tau_{\rm r}} \mathbf{u}_{\tau_{\rm r}}\mathbf{u}_{\tau_{\rm r}}^{\rm H} + (\lambda_{\phi_{\rm r}^{\rm az}} - \chi_{\phi_{\rm r}^{\rm az}\phi_{\rm r}^{\rm el}}) \mathbf{u}_{\phi_{\rm r}^{\rm az}}\mathbf{u}_{\phi_{\rm r}^{\rm az}}^{\rm H} \nonumber \\
	&~~+ (\lambda_{\phi_{\rm r}^{\rm el}}-\chi_{\phi_{\rm r}^{\rm el}\phi_{\rm r}^{\rm az}} )\mathbf{u}_{\phi_{\rm r}^{\rm el}}\mathbf{u}_{\phi_{\rm r}^{\rm el}}^{\rm H}.
\end{align}
Finally, the \eqref{THE1} is proven.

\subsection{Proof of Non-negativity}\label{A3}
We take $\lambda_{\phi_{\rm r}^{\rm az}}-\chi_{\phi_{\rm r}^{\rm az}\phi_{\rm r}^{\rm el}}$ as an example to prove its non-negativity. 
According to \eqref{fim-az2}, \eqref{fim-el2}, and \eqref{fim-azel2}, we can calculate that
\begin{align}
\lambda_{\phi_{\rm r}^{\rm az}}-\chi_{\phi_{\rm r}^{\rm az}\phi_{\rm r}^{\rm el}} & = \kappa_1^2 \big(\mu_{\phi_{\rm r}^{\rm az}\phi_{\rm r}^{\rm az}}\mu_{\phi_{\rm r}^{\rm el}\phi_{\rm r}^{\rm el}}-\mu_{\phi_{\rm r}^{\rm az}\phi_{\rm r}^{\rm el}}\mu_{\phi_{\rm r}^{\rm el}\phi_{\rm r}^{\rm az}}\big)\nonumber\\
& = \kappa_2^2 \cos^2 \!\phi_{\rm r}^{\rm el}(N_{\rm R}^{\rm el}\!-\!1)(7N_{\rm R}\!+\!N_{\rm R}^{\rm az}\!+\!N_{\rm R}^{\rm el}\!-\!5),
\end{align}
where $\kappa_1$ and $\kappa_2$ are real coefficients. Since $N_{\rm R}^{\rm az},N_{\rm R}^{\rm el} \ge 1$, we have $\lambda_{\phi_{\rm r}^{\rm az}}-\chi_{\phi_{\rm r}^{\rm az}\phi_{\rm r}^{\rm el}} \ge 0$. 

\subsection{Proof of Theorem 5}\label{A2}
We derive the EFIM from \eqref{FIM9} as
\begin{equation} \label{EFIM14} 
		\tilde{\mathbf{F}}_{\rm e }(\mathbf{r})\! =\! \mathbf{J} \mathbf{F}(\boldsymbol{\eta})  \mathbf{J}^{\rm H} \!-\! \mathbf{J} \mathbf{F}(\boldsymbol{\eta}) \mathbf{J}_{\rm R}^{\rm  H\!}\! \big(\! \mathbf{J}_{\rm \!R\!} \mathbf{F}(\boldsymbol{\eta}) \mathbf{J}_{\rm R}^{\rm H}+\mathbf{Q}_{\rm R}^{-1}\! \big)^{\!-1}\! \mathbf{J}_{\rm\! R\!} \mathbf{F}(\boldsymbol{\eta}) \mathbf{J}^{\rm H}.		
\end{equation}
Let 
\begin{align}\label{FIM15}
	\tilde{\mathbf{F}}_{\rm A}  =  \diag(\lambda_{\tau_{\rm r}},\lambda_{\phi_{\rm r}^{\rm az}} - \chi_{\phi_{\rm r}^{\rm az}\phi_{\rm r}^{\rm el}},\lambda_{\phi_{\rm r}^{\rm el}}-\chi_{\phi_{\rm r}^{\rm el}\phi_{\rm r}^{\rm az}}),
\end{align}
thus, we have 
\begin{align} \label{EFIM17} 
	\tilde{\mathbf{F}}_{\rm e }(\mathbf{r}) =  \mathbf{U} \tilde{\mathbf{F}}_{\rm A}  \mathbf{U}^{\rm H} -\mathbf{U} \tilde{\mathbf{F}}_{\rm A} \left( \tilde{\mathbf{F}}_{\rm A}+  \tilde{\mathbf{Q}}^{-1}_{\rm R} \right)^{-1}  \tilde{\mathbf{F}}_{\rm A}\mathbf{U}^{\rm H}.
\end{align}
Substituting \eqref{tildeQ} and \eqref{FIM15} into \eqref{EFIM17}, we obtain 
\begin{align} \label{EFIM5} 
	\tilde{\mathbf{F}}_{\rm e }(\mathbf{r}) \! = \!\! \mathbf{U} \!\! \left[\begin{array}{ccc}
		\!\!\!\frac{\lambda_{\tau_{\rm r}}\epsilon^{-1}_{\tau}}{\lambda_{\tau_{\rm r}}+\epsilon^{-1}_{\tau}}  & 0                                  & 0  \\		
		0                      &   \!\!\!\!\frac{(\lambda_{\phi_{\rm r}^{\rm az}}-\chi_{\phi_{\rm r}^{\rm az}\phi_{\rm r}^{\rm el}})\epsilon^{-1}_{\phi^{\rm az}}}{\lambda_{\phi_{\rm r}^{\rm az}}-\chi_{\phi_{\rm r}^{\rm az}\phi_{\rm r}^{\rm el}}+\epsilon^{-1}_{\phi^{\rm az}}}    & 0  \\
		0                      & 0                                  &   \!\!\!\! \frac{(\lambda_{\phi_{\rm r}^{\rm el}}-\chi_{\phi_{\rm r}^{\rm el}\phi_{\rm r}^{\rm az}})\epsilon^{-1}_{\phi^{\rm el}}}{\lambda_{\phi_{\rm r}^{\rm el}}-\chi_{\phi_{\rm r}^{\rm el}\phi_{\rm r}^{\rm az}}+\epsilon^{-1}_{\phi^{\rm el}}}\!\!\!  \\
	\end{array}\right]\!\!\!\mathbf{U}^{\rm H}\!\!.
\end{align}	
Finally, the \eqref{EFIM55} is proven.

{\renewcommand{\baselinestretch}{1}
	\footnotesize
	\bibliographystyle{IEEEtran}
	\bibliography{IEEEabrv,Reference}}

\end{document}